\documentclass[todo=false, usebiblatex=false, pub=false, color=false]{kulieee}                       
\configuretemplate{P. Coppens}{Data-driven distributionally robust MPC for constrained stochastic systems}

\ilarxiv{\onecolumn}
\ilarxiv{\usepackage[textwidth=360pt]{geometry}\setlength{\tikzmarginskip}{4.5cm}}          
\ilpub{\pagestyle{empty}}

\ilarxiv{\title{\Large Data-driven distributionally robust MPC for constrained stochastic systems}}
\ilpub{\title{Data-driven distributionally robust MPC for constrained stochastic systems}}

\author{Peter Coppens and Panagiotis Patrinos$^\dagger$
\thanks{$^\dagger$P. Coppens and P. Patrinos are with the Department of Electrical
Engineering (ESAT-STADIUS), KU Leuven, Kasteelpark Arenberg
10, 3001 Leuven, Belgium.
        {Email: \tt\footnotesize peter.coppens@kuleuven.be, panos.patrinos@kuleuven.be}}%
\thanks{This work was supported by: the Research Foundation
Flanders (FWO) PhD grant 11E5520N and research projects G0A0920N, G086518N and G086318N;
Research Council KU Leuven C1 project No. C14/18/068; Fonds de la Recherche 
Scientifique – FNRS and the FWO – Vlaanderen under 
EOS project no 30468160 (SeLMA); EU's Horizon 2020 research and innovation programme: Marie Skłodowska-Curie grant No. 953348.}%
}%

\usepackage{dutchcal}
\usetikzlibrary{matrix}

\renewcommand{\bar}[1]{\overline{#1}}

\newcommand{\W}{\mathcal{W}}

\newcommand{\smoment}{{C}}
\newcommand{\smomenth}{\hat{\smoment}}
\newcommand{\hm}[1]{\bar{#1}}

\newcommand{\nsample}{M}

\newcommand{\rcn}{\mathcal{r}}

\newcommand{\ra}[1]{\mathbf{r}\text{-}\AVAR^{#1}}
\newcommand{\uLambda}{\mathrm{V}}

\newcommand{\matpart}[1]{[#1]_m}
\newcommand{\vecpart}[1]{[#1]_v}
\newcommand{\cstpart}[1]{[#1]_c}

\begin{document}

\maketitle
\thispagestyle{empty}

\begin{abstract}
In this paper we introduce a novel approach to distributionally robust optimal control that supports online learning
of the ambiguity set, while guaranteeing recursive feasibility. We introduce conic representable risk,
which is useful to derive tractable reformulations of distributionally robust optimization problems.
Specifically, to illustrate the techniques introduced, we utilize risk measures constructed based on data-driven 
ambiguity sets, constraining the second moment of the random disturbance.
In the optimal control setting, such moment-based risk measures lead to tractable optimal controllers when combined with
affine disturbance feedback. Assumptions on the constraints are given that guarantee recursive feasibility.
The resulting control scheme acts as a robust controller when little data is available
and converges to the certainty equivalent controller when a large sample count implies high confidence in the estimated 
second moment. This is illustrated in a numerical experiment. 
\end{abstract}

\begin{pub}
\end{pub}


\section{Introduction}
\IEEEPARstart{\update*{D}}{istributionally} \emph{robust optimization (DRO)} has gained traction recently as a technique that balances 
robustness with performance in an intuitive fashion. From a theoretical point of view such techniques act as 
regularizers \cite{Abadeh2015} and \ilpub{when employed }in a data-driven setting, DRO acts at the interface between 
stochastic and robust optimization \cite{Delage2010}. 

In the control community the potential of such techniques has not gone unnoticed \cite{Kim2020,VanParys2015}. Here one
would ideally solve stochastic optimal control problems like
\begin{equation}\label{eq:socproblem}%
        \begin{alignedat}{2}%
                &\update*{\minimize_{\bm{\pi} \in \bm{\Pi}}} &\quad& \E\left[ 
                        \ssum_{t=0}^{N-1} \ell_t(x_t, \update*{\pi_t(w_0, \dots, w_{t-1})}) + \ell_N(x_N)
                 \right] \nonumber \\                 
               &\update*{\stt}& &x_{t+1} = f(x_t, \update*{\pi_t(w_0, \dots, w_{t-1})}, w_t), \, t \in \N_{0:N-1} \nonumber\\ 
                &&&\prob[\phi(x_t) \leq 0] \geq 1-\varepsilon, \, t \in \N_{1:N-1} \nonumber\\
                &&&\psi(x_N) \leq 0 \text{ a.s.}, \, \update*{x_0 \text{ given},} \nonumber
        \end{alignedat}%
\end{equation}%
where $x_t \in \Re^{n_x}$ denotes the state. Parametrized, causal policies $\bm{\pi} \in \bm{\Pi}$ map disturbances to inputs. That is, an element $\pi_t$ of the sequence $\bm{\pi} = \{\pi_{t}\}_{t=0}^{N-1}$,
maps $\update*{\{w_i\}_{i=0}^{t-1}}$ to inputs in $\Re^{n_u}$ for $t \geq 1$ and $\pi_0 \in \Re^{n_u}$. Here, the disturbances $w_t \in \Re^{n_w}$ are i.i.d. random vectors, the distribution of which is unknown, usually introducing 
the need for robust approaches. DRO then improves upon classical robust control by using available data to infer properties of the distribution, while retaining guarantees.

The core construct in DRO is the \emph{ambiguity set}, a set of distributions against which one should robustify. Several
ambiguity sets have been examined with varying success. The most common are moment-based, $\phi$-divergence and 
Wasserstein ambiguity sets \cite{Rahimian2019}. Such ambiguity sets are connected to so-called \emph{risk measures} by duality. Hence 
this approach is directly related to risk-averse optimization \cite{Pichler2021}.

Throughout the paper we rely on \emph{conic representable ambiguity and risk} to derive tractable problems, similar to the methodology presented in \cite{Chen2019,Sopasakis2019}. 
The main contributions are then as follows: \emph{(i)} we derive tight,
data-driven, moment-based ambiguity sets that are conic representable and shrink when more data becomes available; \emph{(ii)} we extend 
conic risks to the multi-stage setting and use them to model \emph{average value-at-risk} constraints; 
\emph{(iii)} we synthesize the controller such that it is
recursive feasible when it is applied in a receding horizon fashion; \emph{(iv)} we illustrate 
how our framework leads to tractable controllers based on
affine disturbance feedback policies \cite{Goulart2005}, which are evaluated in numerical experiments.

Similar results were achieved in \cite{Mark2020} for a tube-based approach with Wasserstein ambiguity, the radius of which is not data-driven; 
in \cite{Lu2020} for relaxed, robust constraints; 
in \cite{VanParys2015} for moment-based ambiguity which is not data-driven and does not guarantee recursive feasibility;
and \cite{Schuurmans2020, SopasakisP2019RMPC} for discrete distributions. DR control of Markov decision processes with finite state-spaces was also 
considered in \cite{Wiesemann2013}.
Our framework supports online learning of truly data-driven ambiguity sets and risk constraints
within a continuous state-space, while guaranteeing recursive feasibility. 

This section continues with notation and preliminaries. Next \cref{sec:single-stage} introduces conic
and data-driven ambiguity in a single-stage setting and \cref{sec:multi-stage} introduces multi-stage extensions as well as the 
optimal control problem that we want to solve. Then \cref{sec:rec-feasibility} 
shows how to construct a controller such that recursive feasibility is guaranteed. Finally \cref{sec:affine_disturbance} 
illustrates how our\ilpub{ introduced} techniques lead to tractable controllers and contains numerical experiments.
\begin{pub}
        Detailed proofs are deferred to the technical report \citear{} for conciseness.
\end{pub}

\subsection{Notation and preliminaries} \label{sec:notation}
Let $\N$ denote the integers and $\Re$ $(\Re_+)$ the (nonnegative) reals. We denote by $\sym{d}$ the symmetric $d$ by $d$ 
matrices and by $\pd{d}$ ($\psd{d}$) the positive (semi)-definite matrices.
For two matrices of compatible dimensions $X, Y$ we use $[X; Y]$ $\left([X, Y]\right)$ for vertical (horizontal) concatenation. We use $\nrm{\cdot}_2$
to denote the spectral norm (Euclidean norm) for matrices (vectors) and $[\cdot]_+ \dfn \max(0, \cdot)$.
For matrices (or vectors) $X, Y$ 
and cone $\K$, let $X \leqc{\K} Y$ ($X \geqc{\K} Y$) be $Y-X\in \K$ ($X-Y \in \K$). When $\K =\psd{d}$ we use $\sleq$ ($\sgeq$).

Meanwhile, $(X, Y) \dfn [\mathrm{vec}(X); \mathrm{vec}(Y)]$ interprets $X, Y$ as column 
vectors in vertical concatenation. Let $\diag(X, Y)$ be a (block) diagonal matrix
and let $I_d \in \sym{d}$ denote the identity. For a vector $x \in \Re^d$,  $[x]_i$ 
denotes the $i$'th element.

\paragraph{Slice notation} 
We introduce $\N_{a:b}= \left\{ a, \dots, b
 \right\}$. Similarly we use $\bm{w}_{a:b}$ to denote the sequence
 $\left\{ w_i \right\}_{i \in \N_{a:b}}$.
 For a sequence of length $N$, index $a$ ($b$) is omitted when $0$ ($N-1$) 
is implied (when both are omitted we write $\bm{w}$).

Interpreting $\bm{w}_{0:N-1}$ with $w_i \in \Re^{n_w}$ as an element of $\Re^{Nn_w}$,
consider affine maps $\bm{x}_{0:M-1} = \bm{A} \bm{w}_{0:N-1} + \bm{a}_{0:M-1}$
($\bm{A} \in \Re^{Mn_x \times Nn_w}$). Introducing homogeneous
coordinates $\hm{\bm{w}} = (\bm{w}, 1)$ gives $\bm{x} = \hm{\bm{A}} \hm{\bm{w}}$,
with $\hm{\bm{A}} = [\bm{A}, \bm{a}]$. 

For a matrix $\bm{A}$ acting on sequences, the slice $\bm{A}_{i:j, k:\ell}$ describes the part mapping $\bm{w}_{k:\ell}$ 
to $\bm{x}_{i:j}$. So we take block rows and block columns, with blocks of size $\Re^{n_x \times n_w}$.

\paragraph{Risk measures and ambiguity}
Given some measurable space $(\W, \B)$ with $\W$ a compact subset of $\Re^{n_w}$ and $\B$ the associated Borel
sigma-algebra, we use $\Mf(\W)$ ($\Msf(\W)$) to denote the space of finite (signed) measures on $(\W, \B)$,
making the dependency on $\W$ explicit. Similarly,
let $\Pf(\W)$ denote the space of probability measures.

We also consider the space $\Z \dfn \Cont(\W)$
of continuous (bounded) $\B$-measurable functions $z \colon \W \to \Re$. Elements of $\Z$ act as 
random loss functions. Notation $z \sim \mu$ means $z$ has distribution $\mu \in \Pf(\W)$. 
The space $\Msf(\W)$ and $\Z$ are paired by the bilinear form \cite[\S2.2]{Pichler2021},
for $z \in \Z$, $\mu \in \Msf(\W)$,
\begin{equation*}
        \< z, \mu \> \dfn \ilarxiv{\int}\ilpub{\sint}_{\W} z(w) \di \mu(w).
\end{equation*}
We endow $\Msf(\W)$ with the weak$^*$ topology. 

We write $\notation*{z \geqas 0}$ ($\notation*{\mu \geqas 0}$) to imply $z(w) \geq 0$ ($\mu(w) \geq 0$), 
$\forall w \in \W$. Note that, since $\Z$ and $\Msf(\W)$ are linear spaces, we can use the usual
notation of linear operators (e.g., let $E \colon \Msf(\W) \to \Re^{n}$,
then $E \mu = (\<\epsilon_0, \mu\>, \dots \<\epsilon_{n-1}, \mu\>)$ for some random variables $\epsilon_i \update*{\in \Z}$, $i\in\N_{0:n-1}$). 
For each linear operator $E$ we have the adjoint $E^* \colon \Re^{n} \to \Z$,
with $E^*\lambda \dfn \notation*{(\epsilon_0, \dots, \epsilon_{n-1}) \cdot \lambda}$, where $\cdot$ is the usual inner product between vectors.
\begin{arxiv}%
After all,%
\begin{align}%
        E \mu \cdot \lambda &= (\<\epsilon_0, \mu\>, \dots \<\epsilon_{n-1}, \mu\>) \cdot \lambda \nonumber\\
        &= \left(\int_{\W} (\epsilon_0(w), \dots, \epsilon_{n-1}(w)) \di \mu(w) \right) \cdot \lambda\nonumber \\
        &= \int_{\W} \left((\epsilon_0(w), \dots, \epsilon_{n-1}(w)) \cdot \lambda \right) \di \mu(w) \nonumber\\
        &= \<(\epsilon_0, \dots, \epsilon_{n-1}) \cdot \lambda, \mu\> = \<\adj{E}\lambda, \mu\>. \label{eq:E-adj}
\end{align}%
\end{arxiv}%

\begin{prepupdate}
We define risk based on its ambiguity as in most DRO literature \cite{Pichler2021}.
Specifically, we say that $\amb \subseteq \Mf(\W)$ is an \emph{ambiguity set} if it is a non-empty, closed and convex subset of $\Pf(\W)$. 
The associated \ilpub{\emph{coherent}} risk measure $\rme_{\amb} \colon \Z \to \Re$ is then \cite[\S2]{Pichler2021}
\begin{equation} \label{eq:rme}
        \rme_{\amb}[z] = \max_{\mu \in \amb}  \<z, \mu\> = \max_{\mu \in \amb}  \E_{\mu}[z],
\end{equation}
where $\E_{\mu}[\cdot]$ denotes the expected value w.r.t. $\mu \in \Pf(\W)$.
and constitutes a mapping from random loss functions to the real line, 
which (similarly to expectation) can be used to deterministically compare random variables. 
\end{prepupdate}

\begin{arxiv}
Our definition of an ambiguity set is directly related to that of coherent risk \cite{Ruszczynski2006b}.
\begin{lemma} \label{lem:rme}
        Suppose that $\amb \subseteq \Pf(\W)$ is non-empty, closed and convex. Then $\rme_{\amb}$ in \eqref{eq:rme} is
        coherent. 
        Specifically, $\forall z, z' \in \Z$ and $\alpha \in \Re$, $\rme_{\amb}$ is
        \begin{enumerate}[label=(\roman*)]
                \item convex, proper, and lower semi--continuous; \label{lem:coh:prop}
                \item monotonous: $\rme_{\amb}(z) \geq \rme_{\amb}(z')$ if $z \geqas z'$; \label{lem:coh:mon}
                \item translation equivariant: $\rme_{\amb}(z+\alpha) = \rme_{\amb}(z) + \alpha$;\label{lem:coh:tran}
                \item positive homogeneous $\rme_{\amb}(\alpha z) = \alpha \rme_{\amb}(z)$ if $\alpha > 0$.\label{lem:coh:hom}
        \end{enumerate}
        Moreover, $\amb$ is compact and equal to the domain of $\adj{\rme_{\amb}}$ and $\rme_{\amb}[z]$ is finite, where $\adj{\rme_{\amb}}$ 
        denotes the convex conjugate. 
\end{lemma}
\begin{proof}
        Let $\chi_{\amb}$ denote the indicator function of $\amb$ (i.e. $\chi_{\amb}[\mu] = +\infty$ if $\mu \notin \amb$ and $0$ otherwise). Then, by \eqref{eq:rme},
        \begin{equation} \label{eq:rme-conj}
                \rme[z] = \adj{\chi} = \sup_{\mu \in \Msf(\W)} \{ \<z, \mu\> - \chi[\mu]\},
        \end{equation}
        where we omit the subscript of $\rme_{\amb}$ and $\chi_{\amb}$ for convenience. Since $\chi$ is an indicator function, it is convex ($\amb$ is convex);
        lower semi--continuous ($\amb$ is closed); and proper ($\amb$ is nonempty).
        Therefore, by \cite[Prop.2.112]{Bonnans2000} and \eqref{eq:rme-conj}, $\rme$ is proper, convex and lower semi--continuous 
        (its epigraph is an intersection of closed halfspaces). Therefore \ref{lem:coh:prop} holds. 

        Since $\chi$ is convex and lower semi--continuous, we apply \cite[Thm.~2.133]{Bonnans2000} to show $\chi = \chi^{**} = \adj{\rme}$,
        where the second equality follows by \eqref{eq:rme-conj}.
        Hence the domain of $\adj{\rme}$ is $\amb$. 

        Compactness of $\Pf(\W)$ follows by Prohorov's theorem \cite[p.13]{Shapiro2001}. 
        Since $\amb$ is a closed subset of $\Pf(\W)$ it is also compact.

        The results \ref{lem:coh:mon}--\ref{lem:coh:hom} follow directly from \cite[Thm.~2.2]{Ruszczynski2006b}. Specifically \ref{lem:coh:mon}
        from $\amb \subset \Mf(\W)$, \ref{lem:coh:tran} from $\mu(\W) = 1$ for all $\mu \in \amb$ and \ref{lem:coh:hom} from \eqref{eq:rme}. 

        Next, we show that for any $z \in \Z$,
        \begin{equation} \label{eq:robineq}
                \<z, \mu\> \leq \alpha, \, \forall \mu \in \Pf(\W) \quad \Leftrightarrow \quad \notation*{z \leqas \alpha},
        \end{equation}
        where the inequality on the right holds pointwise over $\W$ (i.e. $z(w) \leq \alpha, \forall w \in \W$).
        The argument for \eqref{eq:robineq} is as follows \cite[Eq.~3.7]{Shapiro2001}.
        Since $\mu(\W) = 1$,  $\<z, \mu\> \leq \alpha$ iff $\<\alpha - z, \mu\> \geq 0$, which holds if $\notation*{\alpha - z \geqas 0}$ 
        (since $\notation*{\mu \geqas 0}$). For the converse note $\delta_w \in \Pf(\W)$ 
        for any $w \in \W$ with $\delta_w$ a dirac measure. So  $\<\alpha - z, \delta_w\> = \alpha - z(w) \geq 0$ for any $w \in \W$ is a necessary condition. 
        So we have shown \eqref{eq:robineq}
        
        From \eqref{eq:robineq} we can conclude $\<z, \mu\> \leq \sup_{w\in\W} z(w)$, $\forall z \in \Z, \mu \in \amb$.
        Hence, by \eqref{eq:rme}, $\rme[z] \leq \sup_{w\in\W} z(w)$.
        Since $z(w)$ is finite for any $w \in \W$, $\rme[z]$ is finite (cf. \cite[\S2.2]{Pichler2021}).
\end{proof}
\end{arxiv}

\section{Single-stage problems} \label{sec:single-stage}
Given the dual formulation of a risk measure in \eqref{eq:rme}, it is clear that 
the choice of $\amb$ is a critical design decision. In this section we introduce how ambiguity sets, using moment information,
are derived from data. We also introduce \emph{conic representable risk},
used to derive tractable problems.

\subsection{Data-driven risk}
In DRO the reasoning is usually as follows.
Consider a probability space $(\Omega, \F, \prob)$ and the optimization problem
\begin{equation}
        \begin{alignedat}{2}
                &\minimize_{u \in \mathcal{U}} &\qquad& \E_{\update*{\mu_{\star}}}[f(u, w)], \label{eq:sto}
        \end{alignedat}
\end{equation}
with $u \in \Re^{n_u}$ some decision variable, $f$ some loss function
and $w \colon \Omega \to \W$ a random variable with $\W \subset \Re^{n_{w}}$ the (compact) support
of $w$. The main difficulty in solving the stochastic optimization problem \eqref{eq:sto} 
is that the \emph{distribution} (or \emph{push-forward measure}), $\mu_{\star} \in \Pf(\W)$
defined on the sample space $(\W, \B)$ as 
$\mu_{\star}(O) = \prob[w^{-1}(O)]$
for all $O \in \B$ and with $w^{-1}(O)$ the pre--image of $O$, is unknown.

Hence, instead one introduces an ambiguity set $\amb \subseteq \Pf(\W)$,
which contains $\mu_{\star}$ with some confidence. To do so one can estimate 
some statistic $\theta$ based on data. In the case of $\phi$-divergence \cite{Schuurmans2020} and
Wasserstein ambiguity \cite{Kuhn2019}, this $\theta$ is the empirical distribution, while 
for moment-based ambiguity, $\theta$ encapsulates moment information. We will consider
this final case in \cref{sec:moment-based}. To summarize:
\begin{definition} \label{def:dd-amb}
        Consider random variable $w \colon \Omega \to \W$ with
        distribution $\mu_\star$ and i.i.d. samples $\bm{w}_{0:M-1} \colon \Omega \to \W^M$. 
        Let $\theta \colon \W^M \to \Theta$ denote a statistic for a set $\Theta$ and let 
        $\beta \in \Re$ be some radius\footnote{Some moment-based ambiguity
        set can have multiple radii (cf. \cite{Delage2010}).}. 
        Then a \emph{data-driven ambiguity} 
        $\amb \colon \Theta \times \Re \rightrightarrows \Pf(\W)$
        with confidence $\delta \in (0, 1)$ maps $(\theta(\bm{w}), \beta)$ 
        to an ambiguity set $\amb_{\beta}(\theta(\bm{w})) \subseteq \Pf(\W)$
        such that
        \begin{equation} \label{eq:amb_conf}
                \prob[\mu_\star \in \amb_{\beta}({\theta}(\bm{w}_{0:M-1}))] \geq 1 - \delta.
        \end{equation}
\end{definition}
In \cite{Schuurmans2020} this is referred to as a \emph{learning system}.

Based on \eqref{eq:amb_conf} we minimize $\rme_{\amb_{\beta}(\hat{\theta})}[f(u, w)]$ instead of \eqref{eq:sto}.
The result upper bounds \eqref{eq:sto} with probability at least $1-\delta$. 

\subsection{Moment-based ambiguity} \label{sec:moment-based}
As mentioned before, we focus on the case where $\theta$ 
encapsulates moment information. Such ambiguity sets have the 
advantage that \cite{Kuhn2019} \emph{(i)} they can contain measures with support not limited to
the observed samples (unlike most $\phi$-divergence based sets); \emph{(ii)} 
the radius is estimated with reasonable accuracy based on known information of the distribution
(unlike for Wasserstein-based sets); and \emph{(iii)} problem complexity does not grow with the sample count. 

To ensure that an ambiguity set satisfying \eqref{eq:amb_conf} can be derived, we 
assume that $\W$ is bounded, which is often the case in control applications and is therefore 
the usual assumption in robust control. Other common choices are that 
$w$ is multivariate Gaussian or that it satisfies some 
concentration properties (e.g., sub-Gaussian) \cite{Coppens2019}.
We have:
\begin{lemma} \label{lem:dd-moment}
        Let $\W = \left\{ w \in \Re^{n_w} \colon \nrm*{w}_2 \leq r \right\}$ and $R_w = \diag(I_{n_w}, c r)$ with $c \in \Re$. 
        Assume we have a set of i.i.d. samples $\bm{w}_{0:M-1}$
        of $w \sim \mu_{\star}$ and let
        $\smomenth \dfn \ssum_{i=0}^{M-1} \hm{w}_i \trans{\hm{w}_i}/\nsample$.    
        Then,
        \begin{equation*}
                \amb_{\beta}(\smomenth) = \left\{ \mu \in \Pf(\W) \colon
                 \nrm*{ R_w(\smomenth - \E_{\mu}[\hm{w} \,\trans{\hm{w}}]) \trans{R_w}}_2 \leq \beta \right\},
        \end{equation*}
        satisfies \eqref{eq:amb_conf} when $\beta = 0.5 r^2 (1 + \sqrt{1 + 16 c^2}) \sqrt{2\log(2(n_w + 1)/\delta)/M}$.
\end{lemma}
\begin{proof}
        We use a matrix Hoeffding bound \cite[Thm.~1.3]{Tropp2012} with improved constants. See \ilarxiv{\cref{app:dd-moment-proof}}\ilpub{\cite[App.~A]{Arxiv}} for the full proof.
\end{proof}

\begin{remark}
        In the numerical experiments towards the end of the paper we select $c = 1/4$. This choice results in a relatively simple expression for the radius
        \begin{equation*}
                \beta = 0.5 (1 + \sqrt{2}) r^2 \sqrt{2 \log(2(n_w + 1)/\delta)/\nsample}
        \end{equation*}
        and performed well in experiments. 
\end{remark}

\subsection{Conic-representable ambiguity}
We
introduce \emph{conic representable ambiguity} (similar to the framework in \cite{Sopasakis2019,Chen2019}) 
below and show how such risk is related to robust optimization through conic duality.
\begin{definition} \label{def:conic-risk}
        Consider a compact sample space $\W \subset \Re^{n_w}$
        and $\Z = \Cont(\W)$.
        An ambiguity set $\amb$
        is \emph{conic representable} if, for some
        $E, F \colon \Msf(\W) \to \Re^{n_{b}}$ and $b \in \Re^{n_b}$,        
        \begin{equation*}
                \amb = \left\{\mu \in \update*{\Pf(\W)} \colon \exists \nu \in \Mf(\W), E\mu+F\nu \leqc{\K} b\right\},
        \end{equation*}
        with $\nu$ some auxiliary measure and $\K$ a closed, convex cone. 
        Usually we assume $F = 0$. When $F \neq 0$ we refer to the ambiguity as \emph{$\nu$-conic representable}.
        Similarly we refer to $\rho_{\amb}$, as in \eqref{eq:rme}, as ($\nu$-)conic representable risk (conic for short).
\end{definition}
\begin{pub}%
\update*{The parameters of $\amb$ should be selected such that it is an ambiguity set (i.e., a nonempty, closed and convex subset of $\Pf(\W)$).
In practice $\amb$ will often be chosen as nonempty. Moreover, since $E$ and $F$ are 
continuous mappings over $\Mf$ by construction, closedness of $\amb$ follows from closedness of $\Pf(\W)$ \cite{Arxiv}. 
Convexity follows by linearity of $E$, $F$ and convexity of $\K$.}
\end{pub}%
\begin{arxiv}%
The parameters of $\amb$ should be selected such that it is an ambiguity set (i.e., a nonempty, closed and convex subset of $\Pf(\W)$).
Since we usually want an $\amb$ satisfying \eqref{eq:amb_conf},
it will be non-empty as it should at least contain the true distribution. The random variables used to construct $E$ and $F$,
are all continuous. Therefore \cite[Thm. 15.5]{Aliprantis2006} $E$ and $F$ are continuous mappings. Thus $\amb$ is 
the intersection between the closed set $\Pf(\W)$ and the pre--image of a closed set under a continuous mapping, which is also closed. 
Hence $\amb$ is a closed subset of $\Pf(\W)$. Convexity of $\amb$ then follows, since $E$ and $F$ are linear and $\K$ is convex.
\end{arxiv}%

In \cite{Sopasakis2019} it was shown that both the \emph{average} 
and \emph{entropic} value-at-risk are conic\ilpub{ representable} whenever 
$\W$ is finite. Many more \ilpub{risk measures}\ilarxiv{risks} fall under this framework \cite{Chen2019,Wiesemann2014}.

Direct application of conic linear duality \ilarxiv{\cite{Shapiro2001}}\ilpub{\cite{Shapiro2001,Arxiv}}
gives:
\begin{lemma} \label{lem:cr-dual}
        \update*{A risk $\rme_{\amb}[z]$ as in \cref{def:conic-risk} \ilpub{equals}\ilarxiv{is equal to the optimal value of}}
        \ilpub{\begin{equation}
                \min_{\lambda \geqc{\adj{\K}} 0, \update*{\tau}} \{ \update*{\tau} +  \notation*{b \cdot \lambda} \colon
                \update*{\tau} + \notation*{\adj{E}\lambda \geqas z}, \, \notation*{\adj{F}\lambda \geqas 0} \}, \nonumber
        \end{equation}}\ilarxiv{%
        \begin{equation} \label{eq:rme-dual}
                \begin{alignedat}{2}
                        &\minimize_{\lambda \geqc{\adj{\K}} 0, \tau} &\qquad&  \tau + b \cdot \lambda \\
                        &\stt &&  \adj{E}\lambda + \tau \geqas z, \, \adj{F}\lambda \geqas 0,
                \end{alignedat}\tag{D}
        \end{equation}}
        where the functional inequalities should hold pointwise for all $w \in \W$,
        $\adj{E}$ and $\adj{F}$ denote the adjoint operators \update*{(cf. \cref{sec:notation})},
        and $\adj{\K}$ the dual cone.
\end{lemma}
\begin{arxiv}\begin{prepupdate}%
\begin{proof}
        By \eqref{eq:rme} the primal problem is
        \begin{equation} \label{eq:rme-primal}
                \begin{alignedat}{2}
                &\maximize_{\mu, \nu \in \Mf(\W)} &\qquad& \<z, \mu\> \\
                &\stt && E \mu + F\nu \leqc{\K} b\\
                &&& \<1, \mu\> = 1, 
                \end{alignedat}\tag{$P$}
        \end{equation}
        with $\mathrm{val}(P) = \rme[z]$ (where we omit the subscript for convenience). 
        We refer to the minimization problem \eqref{eq:rme-dual} as the dual problem. 
        Let $\tau \in \Re$ and $\lambda \in \Re^{n_b}$. Then the Lagrangian is
        \begin{align*}
                \varphi[\mu, \nu, \lambda] &\dfn \<z, \mu\> + (1 - \<1, \mu\>)\cdot \tau + (b - E \mu - F \nu) \cdot \lambda \\
                &= \tau + b \cdot \lambda - \<\tau + \adj{E} \lambda - z, \mu\> - \<\adj{F} \lambda, \nu \>,
        \end{align*}
        where we can use \eqref{eq:E-adj} to construct the adjoints. We have
        \begin{equation*}
                \adj{\K} \dfn \{\lambda \in \Re^{n_b} \colon \adj{\lambda} \cdot \lambda \geq 0, \forall \adj{\lambda} \in \K\}.
        \end{equation*}
        Hence $\max_{\nu \in \Mf(\W)} \min_{\lambda \in \adj{\K}, \tau} \{(b - E\mu - F\nu) \cdot \lambda\} = -\chi[\mu]$, where $\chi$ is 
        the indicator of $\amb$. Therefore
        \begin{equation*}
                \max_{\mu, \nu \in \Mf(\W)} \min_{\lambda \in \adj{\K}, \tau} \{\varphi[\mu, \nu, \lambda]\} = \max_{\mu \in \Mf(\W)} \{\<z, \mu\> - \chi[\mu]\} = \rme[z].
        \end{equation*}
        Similarly note that \cite[Eq.~3.7]{Shapiro2001}
        \begin{align*}
                \adj{\Mf}(\W)   &\dfn \{z \in \Z \colon \<z, \mu\> \geq 0, \forall \mu \in \Mf(\W)\} \\
                                &= \{z \in \Z \colon z(w) \geq 0, \forall w \in \W\} = \{z \in \Z \colon z \geqas 0\},
        \end{align*}
        which follows from a similar argument as \eqref{eq:robineq}. As such
        \begin{equation*}
                \min_{\lambda \in \adj{\K}, \tau} \max_{\mu, \nu \in \Mf(\W)} \{\varphi[\mu, \nu, \lambda]\} = \mathrm{val}(D),
        \end{equation*}
        since $\lambda$ gives a finite cost iff $\tau + \adj{E}\lambda - z\in \adj{\Mf}(\W)$ and $\adj{F}\lambda \in \adj{\Mf}(\W)$ 
        (i.e. $\tau + \adj{E} \lambda - z \geqas 0$ and $\adj{F} \lambda \geqas 0$).

        All that is left is to show strong duality (i.e. $\mathrm{val}(D) = \mathrm{val}(P) = \rme[z]$). This follows directly from
        coherence of $\rme$ (specifically $\rme$ being proper, implying consistency of \eqref{eq:rme-primal}), compactness of $\W$ and \cite[Cor.~3.1]{Shapiro2001}. 
\end{proof}        
\end{prepupdate}%
\end{arxiv}%

Note that constraints in the dual are robust constraints, since they hold for all $w \in \W$. Hence,
techniques from robust optimization enable finding tractable
reformulations.
\begin{example} \label{ex:moment-dual}
The ambiguity set ${\amb_{\beta}(\smomenth)}$ of \cref{lem:dd-moment} is conic
with    \update*{$n_b = 3n_w^2$,
        $E\mu = (\pm\langle R_w \hm{w}\, \trans{\hm{w}} \trans{R_w}, \mu \rangle)$, 
        $b = (R_w\hat{C}\trans{R}_w   \pm\beta I)$ and $\K = \psd{n_w+1} \times \psd{n_w+1}$.}

        Moreover, letting \update*{$\lambda = (\Lambda, \uLambda)$} with 
        $\Lambda, \uLambda \in \sym{n_w+1}$ and $\tau \in \Re$
        while using \cref{lem:cr-dual}, means\update*{
\begin{equation*}
        \begin{alignedat}{2}
                \rho_{\amb_{\beta}(\smomenth)}[z] = &\min_{\Lambda, \uLambda \sgeq 0, \tau} &\ilpub{\,}\ilarxiv{\qquad}& \tau 
                        + \tr [\Lambda (R_w \hat{C}\trans{R_w}  +\beta I)] + \tr[\uLambda ( R_w\hat{C}\trans{R}_w  -\beta I)]  \\
                &\sttshort && \tau + \notation*{\adj{E} \lambda \geqas z},                
        \end{alignedat}
\end{equation*}}
where the adjoint $\adj{E} \colon \Re^{n_b} \to \Z$, is\ilarxiv{ (cf. \eqref{eq:E-adj})}
\begin{align*}
        (\tau + \adj{E} \lambda)(w) &= \tau + \tr[R_w \hm{w} \,\trans{\hm{w}} \trans{R_w} (\Lambda - \uLambda)]  \\
                                     &= \trans{\hm{w}} \trans{R_w} (\Lambda - \uLambda) R_w \hm{w} + \tau.
\end{align*}
If 
the constraint $\update*{\tau} + \notation*{\adj{E} \lambda \geqas z}$ is LMI representable,
then $\rho_{\amb_{\beta}(\smomenth)}[z]$ can be evaluated by solving a SDP. 
For example if $z = \trans{\hm{w}} P \hm{w}$. Then, since $\trans{w} w \leq r^2$,
we can apply the S-Lemma \cite[Thm.~B.2.1.]{Ben-Tal2009} to show that $\update*{\tau} + \notation*{\adj{E} \lambda \geqas z}$ iff.,
\begin{align} \label{eq:moment-dual}
       \exists s \geq 0, \, \trans{R_w} (\Lambda - \uLambda) R_w + \diag(sI, \tau - sr^2) - P \sgeq 0.
\end{align}
\end{example}

We also consider ambiguity with only support constraints.
\begin{example} \label{ex:robust-dual}
        Ambiguity $\Pf(\W)$ is conic representable with \update*{$n_b = 0$}.
        Hence,
        $\rho_{\Pf(\W)}[z] = \min_{\tau} \{ \tau \colon \notation*{\tau \geqas z} \}, $
        corresponds to $\rho_{\Pf(\W)}[z] = \max_{w \in \W} z(w)$ and only considers the support as is common in robust optimization.
\end{example}

\section{Multi-stage problems} \label{sec:multi-stage}
In this section we show how conic single-stage risk can be extended to a multi-stage setting,
which is required to develop distributionally robust MPC controllers. 
Specifically, we will consider risk measures operating on the dynamics
\begin{equation*}
        x_{t+1} = f(x_t, u_t, w_t),
\end{equation*}
with $x_t \in \Re^{n_x}$ ($u_t \in \Re^{n_u}$) the state (input) and $w_t \in \Re^{n_w}$
the disturbance, which follows a random process. For $t \in \N_{0:N-1}$ 
we consider $\ell_t \colon \Re^{n_x} \times \Re^{n_u} \to \Re_{+}$ a stage cost
function, and $\ell_N \colon \Re^{n_x} \to \Re_{+}$ the terminal cost.

For each stage $t$, the trajectory up to that time $\bm{w}_{0:t-1}$ is an element of $\W^t$. For 
each $\W^t$, $\B^t$ is the accompanying Borel sigma-algebra, ($\Msf^t$) $\Mf^t$ the set of (signed) measures and $\Pf^t$ 
the set of probability measures on $(\W^t, \B^t$). 
For brevity we henceforth omit the explicit dependency on $\W^t$. Also consider the 
paired spaces of continuous functions $\Z_t = \Cont(\W^t)$. 

We can then consider \update*{multistage ambiguity sets $\amb^t$, which are nonempty, closed and convex subsets of $\Pf^t$.
These in turn define a multistage analog to risk measures\footnote{Multistage risk is often constructed using nested conditional risk measures. 
We avoid such a construction for conciseness and tractability. The consequences of this are discussed in \cite[\S4]{Pichler2021}.}, 
\emph{multistage risk measures} \cite[\S4.2]{Pichler2021}, $\rme_{\amb^t} \colon \Z_t \to \Re$.} Since this is simply a usual risk measure, but defined on $\Z_t$,
the properties of \ilpub{\cref{sec:notation}}\ilarxiv{\cref{lem:rme}} generalize. \shorten*{We specifically consider coherent multistage risk}
\begin{equation} \label{eq:prod-risk-dual}
        \rme_{\amb^t}[z_t] = \max_{\mu^t \in \amb^t} \<z_t, \mu^t\> = \max_{\mu^t \in \amb^t} \E_{\mu^t}[z_t].
\end{equation}

Given such risks, the goal is to solve, \update*{for a given $x_0$,}
\begin{subequations}\label{eq:ocproblem}%
        \begin{alignat}{2}%
                &\update*{\minimize_{\bm{\pi} \in \bm{\Pi}}} &\ilarxiv{\qquad}\ilpub{\quad}& \rme_{\amb^N}\left[ 
                        \sum_{t=0}^{N-1} \ell_t(x_t, \update*{\pi_t(\bm{w}_{:t-1})}) + \ell_N(x_N)
                 \right] \label{eq:ocproblem:cost} \\                 
               &\update*{\stt}& &x_{t+1} = f(x_t, \update*{\pi_t(\bm{w}_{:t-1})}, w_t), \, t \in \N_{0:N-1} \label{eq:ocproblem:a}\\ 
                &&&\rcn^t_{\amb}[\phi(x_t)] \leq 0, \, t \in \N_{1:N-1} \label{eq:ocproblem:b}\\
                &&&\psi(x_N) \leq 0 \text{ a.s.}, \label{eq:ocproblem:c}
        \end{alignat}%
\end{subequations}%
where $\bm{\Pi}$ denotes a set of parametrized, continuous, causal policies.
The risk constraints \eqref{eq:ocproblem:b} involve the multistage risk measures $\rcn^t_{\amb}$ and are discussed in 
detail in \cref{sec:rec-feasibility}. 
We illustrate how \eqref{eq:ocproblem} interpolates between the robust setting and \eqref{eq:socproblem} in
\cref{sec:affine_disturbance}. 
\begin{remark} \label{rem:timecons}
        Problem \eqref{eq:ocproblem} is not exact as we optimize over parametrized policies (cf. \cref{sec:affine_disturbance}), for tractability. 
        As such, time-consistency \cite{Shapiro2012} cannot be guaranteed (i.e. a policy
        computed at $t = 0$ may not be optimal at $t=1$ after realization of $w_0$). 
        Hence a receding horizon scheme is used. 
\end{remark}

\subsection{Product ambiguity} \label{sec:independence}
To enforce independence of the disturbances $w_t$ we introduce \emph{product ambiguity} \cite[\S4.2]{Pichler2021}.
For a sequence of single-stage ambiguity factors $\amb_i$ for $i \in \N_{0:t-1}$, consider
\begin{equation} \label{eq:rm-prod}
        \bigtimes_{i=0}^{t-1} \amb_i = \amb_0 \times \dots \times \amb_{t-1},
\end{equation}
where some $\mu^t \in \amb_0 \times \dots \times \amb_{t-1}$ 
if it is constructed as a product measure of some $\mu_i \in \amb_i$, for $i \in \N_{0:t-1}$ (denoted by $\mu^t = \mu_0 \times \dots \times \mu_{t-1}$). 
We show that in certain cases such ambiguities are conic representable. 

Before doing so we need to extend linear operators $E_i \colon \Msf \to \Re^{n_b}$ to take arguments in $\Msf^t$
in a natural way. To do so, note that for any $E_i \colon \Msf \to \Re^{n_b}$ and $\mu \in \Msf$ we have
$E_i\mu = \sint_{w\in W} e_i(w) \di \mu(w)$ for some $e_i \colon \W \to \Re^{n_b}$, by definition. 
Measures $\mu^t \in \Msf^t$ take arguments $\bm{w}_{:t-1} = (w_0, \ldots, w_{t-1})$, 
so we introduce $E|_i \colon \Msf^t \to \Re^{n_b}$ such that
\begin{equation} \label{eq:conditional-linear}
        E|_i \mu^t = \ilarxiv{\int}\ilpub{\sint}_{\W^t} e_i(w_i) \di \mu^t(\bm{w}_{:t-1}), \, \forall \mu^t \in \Msf^t.
\end{equation}

With these new operators we have
\begin{lemma} \label{lem:cr-dual-prod}
        \update*{Let $\amb_i$ be conic representable with parameters $E_i, b_i, \K_i$ for $i \in \N_{0:t-1}$. 
        Then $\times_{i=0}^{t-1} \amb_i$} is also conic representable
        with parameters
        \begin{align*}
                E\mu^t &=  (E|_0\mu^t, \dots, E|_{t-1}\mu^t),  &
                b &= (b_0, \dots, b_{t-1}),
        \end{align*}
        and $\K = \K_0 \times \dots \K_{t-1}$.
        Moreover, \ilpub{$\rme_{\times_{i=0}^{t-1} \amb_i}[z_t]$ equals}
        \begin{align*}%
                \update*{\ilarxiv{\rme_{\times_{i=0}^{t-1} \amb_i}[z_t] =} \min_{\lambda_i \geqc{\K_i^*} 0, \tau} 
                        \left\{  \tau + \ssum_{i=0}^{t-1} \notation*{b_i \cdot \lambda_i} \colon \tau +  \notation*{\ssum_{i=0}^{t-1} \adj{E|_i} \lambda_i \geqas z_t}\right\}.}
        \end{align*}%
\end{lemma}
\begin{proof}%
\begin{arxiv}%
        Let $\mu^t = \mu_0 \times \dots \times \mu_{t-1}$, with $\mu_i \in \Pf_i$. Then,
        following the notation in \eqref{eq:conditional-linear},
        \begin{align*}
                E|_i\mu^t &\dfn \int_{\W^t} e_i(w_i) \di \mu^t(\di \bm{w}_{0:t-1})\\
                 &\labelrel={step:prodexp} \int_{\W} \dots \int_{\W} e_i(w_i) \di \mu_0(w_0) \dots \di \mu_{t-1}(w_{t-1}) 
                = \int_{\W} e_i(w_i) \di \mu_i(w_i),
        \end{align*}
        where \ref{step:prodexp} follows from $\mu^t = \mu_0 \times \dots \times \mu_{t-1}$ and $\mu^t \in \Pf^t$. 
        Hence $E|_i \mu \leqc{\K_i} b_i$ iff $E_i \mu_i \leqc{\K_i} b_i$. Repeating the same argument 
        for each $i$ proves that $\amb^t \dfn \times_{i=0}^{t-1} \amb_i$ is conic representable. 
        Since $\amb_i$ are all non-empty, $\amb^t$ is also nonempty. Convexity and closedness follow
        from the arguments below \cref{def:conic-risk}. The dual then follows from applying \cref{lem:cr-dual}.
\end{arxiv}%
\ilpub{The proof follows directly from \eqref{eq:conditional-linear} (cf. \cite{Arxiv}).}
\end{proof}%

\subsection{Risk constraints} \label{sec:risk-constraints}
Ideally constraints like \eqref{eq:ocproblem:b} would require the state to lie within some 
set almost surely. Since such a constraint in a stochastic setting can be very conservative,
we will instead implement average value-at-risk constraints, for $\alpha \in (0, 1)$,
\begin{equation} \label{eq:avar}
        \AVAR_{\alpha}^\mu[z] \dfn \inf_{\tau\in\Re} \left\{ \tau + \alpha^{-1} \E_{\mu}[z - \tau]_+ \right\} \leq 0.
\end{equation}
Such constraints \emph{(i)} act as a convex relaxation of chance constraints \cite{Shapiro2009}; \emph{(ii)} penalize the expected violation
in the $\alpha$ quantile where violations do occur. In control applications \eqref{eq:avar} is natural, since it penalizes large violations more.

To evaluate the expectation in \eqref{eq:avar}, true knowledge about the distribution is needed. Hence, we will operate
on the distributionally robust $\AVAR$ constraint instead:
\begin{equation} \label{eq:ravar}
        \ra{\amb}_{\alpha}[z] \dfn \max_{\mu \in \amb} \AVAR_{\alpha}^\mu[z] \leq 0,
\end{equation}
with $\amb$ the \emph{core ambiguity}. 
If $\amb$ satisfies \eqref{eq:amb_conf}, then \eqref{eq:ravar} implies the chance constraint $\prob[z_t \leq 0] \geq 1 - \varepsilon$
holds with $1-\varepsilon \leq (1-\delta)(1-\alpha)$. 
Moreover, whenever $\amb$ is conic, then robust $\AVAR$ is $\nu$-conic.

\newcommand{\core}[1]{#1_{\mathbf{c}}}
\begin{lemma} \label{lem:ravar-conic}
        Let $\amb$ be conic with parameters $\core{E}, \core{b}, \core{\K}$.
        Then $\ra{\amb}_{\alpha}$ in \eqref{eq:ravar}
        is $\nu$-conic with
        \begin{prepupdate}\begin{align*}
                E\mu &= (\core{E}\mu, \<1, \mu\>),  &
                F\nu &= (\core{E}\nu, \<1, \nu\>), &
                b &= (\core{b}, 1)/\alpha,
        \end{align*}\end{prepupdate}
        and $\K = \core{\K} \times \{0\}$. Moreover, $\ra{\rho}_{\alpha}[z]$ equals
        \begin{arxiv}%
        \begin{equation}  \label{eq:ravar-dual}
                \min_{\lambda \geqc{\adj{\core{\K}}} 0, \tau, \core{\tau}} \left\{  
                        \tau + \alpha^{-1} (\core{\tau} + \core{b} \cdot \lambda)
                        \colon \adj{\core{E}}\lambda + \core{\tau} \geqas 0,\, \adj{\core{E}}\lambda + \core{\tau} + \tau \geqas z \right\}.
        \end{equation}%
        \end{arxiv}
        \begin{pub}%
        \begin{equation}  \label{eq:ravar-dual}
                \begin{alignedat}{2}
                        &\min_{\lambda \geqc{\adj{\core{\K}}} 0, \tau, \update*{\core{\tau}}} &\qquad& \tau + \alpha^{-1} (\update*{\core{\tau}} 
                                + \notation*{\core{b} \cdot \lambda}) \\
                        &\sttshort && \notation*{\adj{\core{E}}\lambda + \update*{\core{\tau}} \geqas 0},\, \notation*{\adj{\core{E}}\lambda 
                                + \update*{\core{\tau}} + \tau \geqas z}.                        
                \end{alignedat}
        \end{equation}%
        \end{pub}
\end{lemma}
\begin{proof}%
        \begin{arxiv}%
                This proof generalizes the methodology of \cite{Zymler2013} to arbitrary conic representable risk. 
                First note that
                \begin{align}
                        &\max_{\mu \in \amb} \inf_{\tau\in\Re} \left\{ \tau + \alpha^{-1} \E_{\mu}[z - \tau]_+ \right\} \label{eq:maxmin-ravar}  
                                =\inf_{\tau\in\Re} \left\{ \tau + \alpha^{-1} \notation*{\rho_{\amb}}[z - \tau]_+ \right\}, 
                \end{align}
                \update*{by \cite[Thm.~2.1]{Shapiro2002}.
                Specifically let $\phi(\tau, w) = \tau + \alpha^{-1} [z(w) - \tau]_+$. Then we have \emph{(i)} $\phi(\tau, \cdot) \in \Z$,
                implying that it is $\mu$ integrable and measurable; \emph{(ii)} $\phi(\cdot, w)$ is convex for any $w \in \W$; \emph{(iii)} $\rho_{\notation*{\amb}}[z -\tau]_+$ 
                is finite (\cref{lem:rme}); \emph{(iv)} the set $\amb \subseteq \Pf(\W)$ is compact (\cref{lem:rme});
                \emph{(v)} $\phi(\tau, \cdot)$ is continuous and hence bounded for any $\tau \in \Re$ on $\W$. Under these properties as well as $\amb$ 
                being convex, \cite[Thm.~2.1]{Shapiro2002} states that strong duality holds,
                allowing us to exchange the inf and the max.} 

                Applying \cref{lem:cr-dual} to \notation*{$\rho_{\amb}$} on the r.h.s., results in \eqref{eq:ravar-dual} 
                (where $[\cdot]_+$ produces two separate constraints and $\core{\tau}$, $\core{\lambda}$ act as the Lagrangian multipliers 
                for the constraint $\core{\mu} \in \core{\amb}$).
                Again applying \cref{lem:cr-dual} gives the original $\nu$-conic representation. 
                
                The second application of \cref{lem:cr-dual} requires the resulting set of measures (denoted $\amb_{\alpha}$ below)
                to be a nonempty, closed and convex subset of $\Pf(\W)$. 
                By construction we already have $\amb_{\alpha} \subseteq \Pf(\W)$. 
                Next, we show that $\amb_{\alpha}$ is larger than $\amb$. After all, for any $\core{\mu} \in \amb$, 
                take $\mu = \core{\mu}$ and $\nu = \alpha^{-1}(1-\alpha)\core{\mu} \geqas 0$, since $\alpha \in (0, 1)$. Moreover, since $\core{\K}$ is a cone,
                \begin{equation*}
                        \core{E}\mu + \core{E}\nu \leqc{\core{\K}} \core{b}/\alpha \quad \Leftrightarrow \quad 
                        \core{E}(\alpha \mu + \alpha \nu) \leqc{\core{\K}} \core{b},
                \end{equation*}
                with $\alpha\mu + \alpha\nu = \alpha \core{\mu} + (1-\alpha)\core{\mu} = \core{\mu}$. For the same reason we have 
                $\<1, \mu\> + \<1, \nu\> = \alpha^{-1} \<1, \core{\mu}\> = \alpha^{-1}$.
                Therefore $\core{\mu} \in \amb_{\alpha}$ for each $\core{\mu} \in \amb$. Hence, since $\amb$ is nonempty, so is $\amb_{\alpha}$. 
                Closedness and convexity then follow by the arguments below \cref{def:conic-risk}. So using \cref{lem:cr-dual} is justified. 
        \end{arxiv}%
        \begin{pub}%
                This proof generalizes \cite{Zymler2013} to arbitrary conic representable risk and uses \cite[Thm.~2.1]{Shapiro2002} (cf. \cite{Arxiv}).        
        \end{pub}%
\end{proof}

\section{Recursive feasibility} \label{sec:rec-feasibility}
We show how one can configure the constraints of \eqref{eq:ocproblem} such that recursive feasibility 
is ensured. To do so we assume
\begin{enumerate}[label=(A\arabic*), ref=(A\arabic*)]
        \item \label{asm:risk} $\mathcal{r}^t_{\notation*{\amb}}[z_t] \dfn \ra{{\Pf^{t-1} \times \amb}}_{\alpha}[z_t]$, $\forall t \in \N_{0:N-1}, z_t \in \Z^t$;
        \item \label{asm:amb} $\amb$ is updated based on measurements as $g(\amb, w)$ (e.g., following \cref{lem:dd-moment}) and $\amb^+ \dfn g(\amb, w) \subseteq \amb$ a.s.
\end{enumerate}

We introduce the terminal set $\mathcal{X}_N \dfn \left\{ x \colon \psi(x) \leq  0 \right\}$.
Let $\mathcal{V}^{\amb}_N(x_0)$ denote the minimum of \eqref{eq:ocproblem} for some $x_0$ and let $\mathcal{D}_N(\amb)$
denote its domain. Then consider the set of feasible policies 
$
        \bm{\Pi}_N(x_0, \amb) \dfn \left\{ 
                \bm{\pi} \in \bm{\Pi} \colon \eqref{eq:ocproblem:a}, \eqref{eq:ocproblem:b}, \eqref{eq:ocproblem:c}
         \right\}.
$

We begin with the following definition.
\begin{definition}[Recursive Feasibility] \label{def:rfeas}
        Let $x_0 \in \mathcal{D}_N(\amb)$ and $\bm{\pi} \in \bm{\Pi}_N(x_0, \amb)$. If,
        $f(x_0, \pi_0, w_0) \in \mathcal{D}_N(g(\amb, w_0))$ a.s., then \eqref{eq:ocproblem} is \emph{recursive feasible} (RF).        
\end{definition}

We can then prove the following theorem:
\begin{theorem} \label{thm:rf}
        Assume \cref{asm:risk}, \cref{asm:amb} and that we are given some terminal policy $\pi_f(x_N)$ such that
        \begin{enumerate}[label=(A\arabic*), ref=(A\arabic*), resume] 
                \setcounter{enumi}{2}
                \item \label{asm:term} $\mathcal{X}_N \subseteq \{x \in \Re^{n_x} \colon \phi(f(x, \pi_f(x), w)) \leq 0, \, \forall w \in \W\}$;
                \item \label{asm:rci} $\mathcal{X}_N$ is robust positive invariant (RPI) for $\pi_f$ (i.e., $f(x, \pi_f(x), w) \in \mathcal{X}_N$
                for each $(x, w) \in \mathcal{X}_N \times \W$);
                \item \label{asm:param} $\forall \bm{\pi}_{0:N-1} \in \bm{\Pi}$, let $\pi_N = \pi_f(x_N)$,
                depending on $\bm{w}_{:N}$ through $x_N$. Then the shifted policy $\bm{\pi}^+_{0:N-1} = \bm{\pi}_{1:N}(w_0, \cdot)$ for any fixed $w_0 \in \W$,
                lies in $\bm{\Pi}$.
        \end{enumerate}
        Then, \eqref{eq:ocproblem} is recursive feasible. 
\end{theorem}
\begin{proof}%
\begin{arxiv}%
        We will consider any fixed $w_0 \in \W$ and show that, given that \eqref{eq:ocproblem} 
        is feasible for some $x_0$, it will also be feasible for the next time step starting from $x_0^+ = f(x_0, \pi_0, w_0)$ (cf.~\cref{def:rfeas}).
        Here, we consider the feasible policy $\bm{\pi}_{0:N-1} \in \bm{\Pi}_N(x_0, \amb)$, to which we append $\pi_N = \pi_f(x_N)$.
        Propagating the dynamics with this policy gives the sequence of states $\bm{x}_{0:N+1}$, depending on $\bm{w}_{:N}$ through \eqref{eq:ocproblem:a}.
        
        We then define the shifted sequence of states as
        \begin{align*}
                \bm{x}_{0:N}^+(\bm{w}_{1:N-1}) \dfn (\bm{x}_1(w_0, \bm{w}_{1:N}), \dots, \bm{x}_{N+1}(w_0, \bm{w}_{1:N})),
        \end{align*}
        where $w_0$ is considered fixed and $\bm{w}_{1:N}$ is left variable. 
        We can analogously define the shifted policy $\bm{\pi}^+_{0:N-1}$ as 
        \begin{equation*}
                \bm{\pi}_{0:N-1}^+(\bm{w}_{1:N-1}) \dfn (\bm{\pi}_1(w_0, \bm{w}_{1:N-1}), \dots, \bm{\pi}_{f}(x_N(w_0, \bm{w}_{1:N-1}))).
        \end{equation*}
        By construction, these shifted sequences satisfy \eqref{eq:ocproblem:a} and we can consider risk measures over (continuous) functions of these, where integration
        is performed over $\bm{w}_{1:N}$.

        Using this coupling between the feasible problem and the shifted problem we show $\bm{\pi}^+ \in \bm{\Pi}_N(x_0, \amb_+)$. That is, the candidate policy $\bm{\pi}^+$
        is feasible for the shifted problem.

        \begin{enumerate}[label=\Roman*:, ref=\Roman*]
                \item By \ref{asm:param}, $\bm{\pi}^+ \in \bm{\Pi}$;
                \item We show that \eqref{eq:ocproblem:b} at $t$ implies \eqref{eq:ocproblem:b} in the shifted problem at $t-1$. 
                That is, $\rcn^t_{\amb}[\phi(x_t)] \geq \rcn^{t-1}_{\amb}[\phi(x_{t-1}^+)]$ for any $w_0 \in \W$. So, letting $z = \phi(x_t)$, by \eqref{eq:maxmin-ravar},
                \begin{align*}
                        \rcn^t_{\amb}[z] &= \max_{\mu^t \in \Pf^{t-1} \times \amb} \inf_{\tau \in \Re} \left\{ \tau + \alpha^{-1} \E_{\mu^t}[z_{\tau}] \right\} 
                                        = \inf_{\tau \in \Re} \left\{ \tau + \alpha^{-1} \rho_{\Pf^{t-1} \times \amb}[z_{\tau}] \right\},
                \end{align*}
                with $z_{\tau} = [z - \tau]_+$. 
                For $z^+ \dfn \phi(x_{t-1}^+) = z(w_0, \bm{w}_{1:t-1})$ ($z^+_{\tau} = [z^+ - \tau]_+$),  we replace $\rho_{\Pf^{t-1} \times \amb}[z_{\tau}]$ 
                with $\rho_{\Pf^{t-2} \times \amb}[z_{\tau}^+]$ 
                for $\rcn^{t-1}_{\amb}[z^+]$. 
                Writing out $\rho_{\Pf^{t-1} \times \amb}$ gives
                \begin{subequations}\label{eq:integrals}
                        \begin{align} 
                                & \max_{\mu^{t} \in \Pf^{t-1} \times \amb} \left\{ \int_{\W^t} z_{\tau} (\bm{w}_{:t-1}) \di \mu^t(\bm{w}_{:t-1}) \right\}\nonumber\\[1em]
                                                &\labelrel={eq:integrals:b} \max_{\mu_{t-1} \in \amb} \max_{\mu^{t-1} \in \Pf^{t-1}} \left\{ \int_{\W^{t-1}} \smashoverbracket{\left(
                                                        \int_{\W} z_{\tau}(\bm{w}_{:t-1})  \di \mu_{t-1}(w_{t-1}) \right)
                                                }{h(\bm{w}_{:t-2})} \di \mu^{t-1}(\bm{w}_{:t-2}) \right\}\nonumber \\
                                                &\labelrel={eq:integrals:c} \max_{\mu_{t-1} \in \amb} \max_{\bm{w}_{:t-2} \in \W^{t-1}} h(\bm{w}_{:t-2}) 
                                                \labelrel\geq{eq:integrals:d} \max_{\mu_{t-1} \in \amb} \max_{\bm{w}_{1:t-2} \in \W^{t-2}} h(w_0, \bm{w}_{1:t-2}).\nonumber
                        \end{align}
                \end{subequations}             
                Noting that $\mu^t = \mu^{t-1} \times \mu_{t-1}$ with $\mu^{t-1} \in \Pf^{t-1}$ and $\mu_{t-1} \in \amb$, before splitting up the max and the integrals,
                gives \ref{eq:integrals:b}. The inner integral (i.e. $h(\bm{w}_{:t-2})$) then acts as a continuous random variable  $\W^{t-1} \to \Re$ for 
                any fixed $\mu_{t-1}$ (cf. App.~\ref{app:continuity}). 
                Hence we can apply the reasoning within \cref{ex:robust-dual} to maximize over $\bm{w}_{:t-2} \in \W^{t-1}$ instead of over measures
                resulting in \ref{eq:integrals:c}. It is clear that, fixing the value of $w_0$ results in the inequality \ref{eq:integrals:d}.
                Reverting the steps \ref{eq:integrals:b} and \ref{eq:integrals:c} to get a maximization over $\mu^{t-1} \in \Pf^{t-2} \times \amb$ shows 
                that the final expression after \ref{eq:integrals:d} equals $ \rho^+_{\amb^{t-1}}[z^+ - \tau]_+$. Hence
                $\rho_{\Pf^{t-1} \times \amb}[z-\tau]_+ \geq \rho_{\Pf^{t-2} \times \amb}[z^+ - \tau]_+$ for all $\tau \in \Re$, $z \in \Z^t$ and $w_0 \in \W$. Therefore 
                $\rcn^t_{\amb}[\phi(x_t)] \geq \rcn^{t-1}_{\amb}[\phi(x_{t-1}^+)]$. Since $\amb^+ \subseteq \amb$ by \ref{asm:amb}, 
                $\rcn^t_{\amb}[\phi(x_t)] \geq \rcn^{t-1}_{\amb^+}[\phi(x_{t-1}^+)]$. Hence \eqref{eq:ocproblem:b} holds for all $t \in \N_{1:N-2}$
                in the shifted problem. For $t = N-1$ we rely on \ref{asm:term} and \ref{asm:rci}.

                \item The terminal constraint \eqref{eq:ocproblem:c} follows directly from \ref{asm:rci}. 
        \end{enumerate}
        We have thus shown that $\bm{\pi}^+$ is a feasible policy.
\end{arxiv}%
\ilpub{The proof follows the usual recursive feasibility proofs from robust control (cf. \cite{Arxiv}).}
\end{proof}
\begin{remark}
        Note that \cref{asm:risk} is essential since RF is a robust property, holding a.s. It acts as a convex relaxation of chance constraints 
        conditioned on previous time steps (i.e., $\prob[\phi(x_t) \leq 0 \mid x_{t-1}] \leq \varepsilon$
        which holds a.s., hence $\forall \mu^{t-1} \in \Pf^{t-1}$ by \cref{ex:robust-dual}). 
        Due to the reduction of the policy space $\bm{\Pi}$ (cf. \cref{rem:timecons}) it is harder to satisfy such constraints for larger $t$. 
        Other (less conservative) reformulations exist in the stochastic MPC literature, which 
        impose all constraints at the first time step using a (maximal) RPI set (cf. \cite{Lorenzen2017}).
\end{remark}

\section{Affine disturbance feedback} \label{sec:affine_disturbance}
To make the reformulations above more concrete, we show how \eqref{eq:ocproblem}
is solved. In general this is intractable, since we need to optimize over
infinite dimensional policies $\bm{\pi}$, under robust constraints associated with the 
risks (cf. \cref{rem:timecons}). Hence, we use affine disturbance feedback. The resulting optimization problem is a SDP. 
Different ambiguity sets and policies would give other reformulations (e.g., \cite{Lu2020}, \cite{Sopasakis2019}). 

Consider linear dynamics, quadratic losses and constraints:
\begin{align*}
        f(x, u, w) &= Ax + Bu + Ew, & \pi_f(x) &= K_f x, \\
        \ell_t(x, u) &= \trans{x} Q x + \trans{u} R u, & 
        \ell_N(x) &= \trans{x} Q_f x, \\
        \phi(x) &=     \trans{x} G x + 2 \trans{g} x + \gamma, \ilpub{\span\span\span\\}\ilarxiv{&}
        \psi(x) &= \trans{x} G_{f} x + 2 \trans{g_{f}} x + \gamma_{f}. \ilpub{\span\span\span}
\end{align*}
with $Q \sgeq 0$, $R \sgeq 0$, $Q_f \sgeq 0$, $G \sgeq 0$, $G_{f} \sgeq 0$. 
We could include (hard) input constraints as well or multiple state constraints (cf. discussion in \cite[\S1]{VanParys2015} on modeling 
joint chance constraints), but abstain from doing so for conciseness.

In this setting affine disturbance feedback \cite{Goulart2005} has been applied to solve many robust optimal control problems
(and even some DRO problems \cite{VanParys2015}). The idea is to let
\begin{equation*}
        \bm{\pi}(\bm{w}) = \bm{F} \bm{w} + \bm{f},
\end{equation*}
where \ilpub{mapping }$\bm{F} \colon \Re^{(N+1)n_w} \to \Re^{(N+1)n_x}$ (defined in \ilarxiv{\cref{app:aff-dist-not}}\ilpub{\cite[\update*{App.~C}]{Arxiv}}), 
has a structure that enforces causality of $\bm{\pi}$. Note $\bm{x}_{:N} \in \Re^{(N+1)n_x}$, $\bm{w}_{:N-1} \in \Re^{Nn_w}$ and $\bm{u}_{:N-1} \in \Re^{Nn_u}$.

The state trajectory then depends on the disturbance as
\begin{equation*}
        \bm{x} = (\bm{B} \bm{F} + \bm{E}) \bm{w} + (\bm{A} x_0 + \bm{B} \bm{f})
         = \hm{\bm{H}} \hm{\bm{w}},
\end{equation*}
with $\bm{A}, \bm{B}, \bm{E}$ defined in \ilarxiv{\cref{app:aff-dist-not}}\ilpub{\cite[\update*{App.~C}]{Arxiv}}, $\hm{\bm{F}} = [\bm{F}, \bm{f}]$
and $\hm{\bm{H}} = [\bm{H}, \bm{{h}}] = [\bm{B} \bm{F} + \bm{E}, \bm{A} x_0 + \bm{B} \bm{f}]$. 
Here the linear part, $\bm{H}$, can be interpreted as the sensitivity of the state to the disturbance, while $\bm{{h}} = ({h}_0, \dots, {h}_N)$ 
is the deterministic part of the state (i.e., the state trajectory when $\bm{w} = 0$).

We first show how the cost \eqref{eq:ocproblem:cost} is implemented, before doing the same for the
risk constraints \eqref{eq:ocproblem:b} and the terminal constraint \eqref{eq:ocproblem:c}.
\ilpub{For conciseness we gloss over classical robust optimization techniques. 
The complete reformulation is given in \cite[\S5]{Arxiv}.}

\subsection{Reformulation of the cost} \label{sec:risk-cost-tractable}
\begin{pub}%
We will assume the risk in \cref{eq:ocproblem:cost} has a product ambiguity,
$\amb^N = \amb_{\beta}(\smomenth) \times \dots \times \amb_{\beta}(\smomenth)$
with $\amb_{\beta}(\smomenth)$ as in \cref{lem:dd-moment}. Letting $z \dfn \trans{\bm{x}} \bm{Q} \bm{x} + \trans{\bm{u}} \bm{R} \bm{u}$, \eqref{eq:ocproblem:cost} is
\begin{align}
        \rme_{\amb^N}[z] &= \rme_{\amb^N} \left[
                \trans{\hm{\bm{w}}} \left(\trans{\hm{\bm{H}}} \bm{Q} \hm{\bm{H}} + \trans{\hm{\bm{F}}} \bm{R} \hm{\bm{F}}\right) \hm{\bm{w}}
        \right]\label{eq:cost-primal}
\end{align}
where $\bm{Q}$ and $\bm{R}$ are defined in \ilarxiv{\cref{app:aff-dist-not}}\ilpub{\cite[\update*{App.~C}]{Arxiv}}.
Note that the argument of $\rho_{\amb^N}$ is quadratic in $\bm{w}$. we
first apply \cref{lem:cr-dual-prod} to dualize \eqref{eq:cost-primal}. We then apply \cref{ex:moment-dual} to the ambiguity factors.
\update*{Let $\Lambda^\rme_i, \uLambda^\rme_i \sgeq 0$, $\tau^\rme \in \Re$, $\lambda_i^{\rme} = (\Lambda^\rme_i, \uLambda^\rme_i)$ and 
$b = (\hat{C} \pm\beta I)$ (as in \cref{ex:moment-dual})}
for all $i \in \N_{0:N-1}$. Then
\begin{alignat}{2}
        \rme_{\amb^N}[z] = &\min_{\Lambda_i^\rho, \uLambda_i^\rho \sgeq 0, \update*{\tau^\rho}} &\quad& \update*{\tau^\rho} + \ssum_{i=0}^{N-1} \notation*{b \cdot \lambda_i^{\rme}} \nonumber \\
        &\sttshort && \update*{\tau^\rho} + \notation*{\ssum_{i=0}^{N-1} \adj{E|_i} \lambda_i^{\rme} \geqas z}, \label{eq:sdp-cost}
\end{alignat}
with $\adj{E|_i} \lambda_i^{\rho} = \update*{\trans{\hm{w}_i} (\Lambda^\rme_i - \uLambda^\rme_i) \hm{w}_i}$ for $i \in \N_{0:N-1}$.
Therefore (cf. $z$ in \cref{eq:cost-primal}), inequality \eqref{eq:sdp-cost} is quadratic in $\bm{w}$. It should hold for all $\bm{w}$ 
in the set of quadratic inequalities $\W^N = \{\bm{w}  \colon \trans{w_i} w_i \leq r^2, \, \forall i \in \N_{0:N-1}\}$. 
This is conservatively approximated by a LMI, with the approximate S-Lemma \cite[Thm.~B.3.1.]{Ben-Tal2009}. 
Its effect here is conservatively quantified as increasing the radius $r$ by a factor $9.19 \sqrt{\ln(N-1)}$, before solving the problem exactly. 
Note that, the \update*{Hessian} of $z$ in \eqref{eq:cost-primal} (cf. $P$ in \cref{ex:moment-dual}) is not linear in $\bm{F}$ and $\bm{f}$. 
So a Schur complement is needed.
\end{pub}%
\begin{arxiv}%
        We will assume the risk in \cref{eq:ocproblem:cost} has a product ambiguity,
        $\amb^N = \amb_{\beta}(\smomenth) \times \dots \times \amb_{\beta}(\smomenth)$
        with $\amb_{\beta}(\smomenth)$ as in \cref{lem:dd-moment}. Letting $z \dfn \trans{\bm{x}} \bm{Q} \bm{x} + \trans{\bm{u}} \bm{R} \bm{u}$, 
        \eqref{eq:ocproblem:cost} is
        \begin{align}
                \rme_{\amb^N}[z] &= \rme_{\amb^N} \left[
                        \trans{\hm{\bm{w}}} \left(\trans{\hm{\bm{H}}} \bm{Q} \hm{\bm{H}} + \trans{\hm{\bm{F}}} \bm{R} \hm{\bm{F}}\right) \hm{\bm{w}}
                \right]\label{eq:rme-cost-primal}
        \end{align}
        where $\bm{Q}$ and $\bm{R}$ are defined in \cref{app:aff-dist-not}.

        Note, by \cref{ex:moment-dual}, that the factors $\amb_i = \amb_{\beta}(\smomenth)$ for $i\in\N_{0:N-1}$ are conic with
        \begin{align*}
                E_i^{\rme}\mu &= (\pm \< R_w \hm{w} \, \trans{\hm{w}} \trans{R_w}, \mu \>), & b^{\rme}_i &= (R_w \smomenth \trans{R_w} \pm \beta I), & 
                \K^{\rme}_i &= \psd{n_w+1} \times \psd{n_w+1}.
        \end{align*}
        Therefore, by \cref{lem:cr-dual-prod}, $\amb^N = \times_{i=0}^{N-1} \amb_i$ is conic with
        \begin{align*}
                E^{\rme}\mu^N &= (E^{\rme}|_0 \mu^N, \dots, E^{\rme}|_{N-1}\mu^N), & b^{\rme} &= (b^{\rme}_0, \dots, b^{\rme}_{N-1})
        \end{align*}        
        and $\K^{\rme} = \K^{\rme}_0 \times \dots \times \K^{\rme}_{N-1}$.
        Let $\Lambda^\rme_i, \uLambda^\rme_i \sgeq 0$, $\tau \in \Re$, $\lambda_i^{\rme} = (\Lambda^\rme_i, \uLambda^\rme_i)$ and $b = (R_w \hat{C} \trans{R_w} \pm\beta I)$
        for all $i \in \N_{0:N-1}$. Then, by \cref{lem:cr-dual-prod}, 
        \begin{subequations} \label{eq:rme-primal-ref}
                \begin{alignat}{2}
                \rme_{\amb^N}[z]=&\min_{\Lambda_i^\rho \sgeq 0, \uLambda_i^\rho \sgeq 0, \tau} &\qquad& \tau 
                        + \ssum_{i=0}^{N-1} b \cdot \lambda_i^{\rme} \label{eq:rme-primal-ref:a} \\
                & \sttshort && \tau + \notation*{\ssum_{i=0}^{N-1} \adj{E^{\rme}|_i} \lambda_i^{\rme} \geqas z}, \label{eq:rme-primal-ref:b}
                \end{alignat}
        \end{subequations}
        with $\adj{E^{\rme}|_i} \lambda_i^{\rho} = \trans{\hm{w}_i} \trans{R_w} (\Lambda^\rme_i - \uLambda^\rme_i) R_w \hm{w}_i$ for $i \in \N_{0:N-1}$.

        Our goal is now to find a tight, conservative LMI reformulation of \eqref{eq:rme-primal-ref:b} as an LMI.
        We use, for any matrix $\Delta \in \sym{n_w+1}$, the partitioning
        \begin{equation*}
                \Delta = \begin{bmatrix}
                        \matpart{\Delta} & \vecpart{\Delta} \\ \trans{\vecpart{\Delta}} & \cstpart{\Delta}
                \end{bmatrix},
        \end{equation*}
        with $\matpart{\Delta} \in \sym{n_w}$, $\vecpart{\Delta} \in \Re^{n_w}$ and $\cstpart{\Delta} \in \Re$. 
        We also introduce $\Delta^\rho_i \dfn \trans{R_w} (\Lambda_i^\rho - \uLambda_i^\rho) R_w$ 
        for all $i \in \N_{0:N-1}$. Then,
        \begin{equation*}
                \sum_{i=0}^{N-1} \adj{E^\rme|_i} \lambda_i^\rho = 
                \trans{\hm{\bm{w}}} \begin{bmatrix}
                        \matpart{\Delta_0^\rho}  &  &  & \vecpart{\Delta_0^\rho} \\
                         & \ddots & & \vdots\\
                         & & \matpart{\Delta_{N-1}^\rho} &\vecpart{\Delta_{N-1}^\rho} \\
                        \trans{\vecpart{\Delta_0^\rho}} & \dots & \trans{\vecpart{\Delta_{N-1}^\rho}} & \tau + \sum_{i=0}^{N-1} \cstpart{\Delta_i^\rho}
                \end{bmatrix} \hm{\bm{w}},
        \end{equation*}
        which is quadratic in $\bm{w}$. 
        Note that the argument of $\rho_{\amb^N}$ is also quadratic in $\bm{w}$.
        Therefore (cf. $z$ in \cref{eq:rme-cost-primal}), \eqref{eq:rme-primal-ref:b} is quadratic in $\bm{w}$. It should hold for all $\bm{w}$ in the
        set of quadratic inequalities $\W^N = \{\bm{w}  \colon \trans{w_i} w_i \leq r^2, \, \forall i \in \N_{0:N-1}\}$. 
        This is conservatively approximated by a LMI, with the approximate S-Lemma \cite[Thm.~B.3.1.]{Ben-Tal2009}. 
        Its effect here is conservatively quantified as increasing the radius $r$ by a factor $9.19 \sqrt{\ln(N-1)}$, before solving the problem exactly. 
        Specifically \eqref{eq:rme-primal-ref:b} holds if there exists a $P^{\rme} \in \psd{N_w + 1}$ and $s^{\rme} \in \Re_+^{N-1}$ such that 
        (remembering $\Delta^\rho_i \dfn \trans{R_w} (\Lambda_i^\rho - \uLambda_i^\rho) R_w$)
        \begin{equation} \label{eq:slemma-cost}
                \begin{bmatrix}
                        \matpart{\Delta_0^\rho} + [s^\rho]_0 I &  &  & \vecpart{\Delta_0^\rho} \\
                         &  \ddots & & \vdots\\
                         & & \matpart{\Delta_{N-1}^\rho} + [s^\rho]_{N-1} I &\vecpart{\Delta_{N-1}^\rho} \\
                        \trans{\vecpart{\Delta_0^\rho}} & \dots & \trans{\vecpart{\Delta_{N-1}^\rho}} & \tau + \ssum_{i=0}^{N-1} \cstpart{\Delta_i^\rho} - [s^\rho]_i r^2
                \end{bmatrix} \sgeq P^\rho 
        \end{equation}
        and
        \begin{equation}
                \begin{bmatrix}
                        P^\rho & \trans{\hm{\bm{H}}} \bm{Q}^{1/2} & \trans{\hm{\bm{F}}} \bm{R}^{1/2} \\
                        \bm{Q}^{1/2} \hm{\bm{H}} & I \\
                        \bm{R}^{1/2} \hm{\bm{F}} & & I
                \end{bmatrix} \sgeq 0. \label{eq:schur-cost}
        \end{equation}
        Here \eqref{eq:schur-cost} implies $\notation*{\trans{\hm{\bm{w}}} P^\rho \hm{\bm{w}} \geq z(\bm{w})}$ for all $\bm{w} \in \Re^{N_w}$ by 
        a Schur complement argument. The cost \eqref{eq:rme-primal-ref:a} meanwhile is given as
        \begin{equation} \label{eq:cost-cost}
                \tau + \ssum_{i=0}^{N-1} \tr[\Lambda^\rho_i(R_w \hat{C} \trans{R_w} + \beta I)] + \tr[\uLambda^\rho_i(R_w \hat{C} \trans{R_w} - \beta I)]. 
        \end{equation}
\end{arxiv}

\subsection{Reformulation of risk constraints} \label{sec:risk-constraint-tractable}
\begin{pub}%
We consider constraints, $\mathcal{r}^t_{\amb}$ with ambiguity $\amb = \amb_{\beta}(\smomenth)$, 
satisfying \cref{asm:risk}. The random variable $\phi(x_t)$ in \eqref{eq:ocproblem:b} equals
\begin{equation} \label{eq:risk-constraint-primal-var}
        \trans{\hm{x}}_t \begin{bmatrix}
                G & g \\ \trans{g} & \gamma
        \end{bmatrix} \hm{x}_t, \quad \text{with } \hm{x}_t = \begin{bmatrix}
                \bm{H}_{t, :t-1} & {h}_t \\ 0 & 1
        \end{bmatrix} \bm{\hm{w}}_{:t-1}
\end{equation}
and is thus quadratic in $\bm{w}_{:t-1}$. By construction, \eqref{eq:ocproblem:b} is of type \eqref{eq:ravar} with core ambiguity $\Pf^{t-1} \times \amb$.
So, by \cref{lem:ravar-conic} followed by \cref{lem:cr-dual-prod}, \eqref{eq:ocproblem:b} is equivalent to:

\update*{There exists a
$\tau_{\bm{c}, t}^{\rcn}, \tau_{t}^{\rcn}\in \Re$,
$\Lambda^{\rcn}_{t}, \uLambda^{\rcn}_{t} \sgeq 0$ and $\lambda^{\rcn}_{t} = (\Lambda^{\rcn}_{t}, \uLambda^{\rcn}_{t})$ such that
\begin{subequations} \label{eq:sdp-cnst-ref}
        \begin{alignat}{1}
        \tau^{\rcn}_{t} + \alpha^{-1} \left(\tau^{\rcn}_{\bm{c},t} + b \cdot \lambda^{\rcn}_t \right) &\leq 0 \label{eq:sdp-cnst-ref:a},\\
        \tau^{\rcn}_{\bm{c},t} + \adj{E^{\rcn}|_{t-1}} \lambda^{\rcn}_{t} &\geqas 0\label{eq:sdp-cnst-ref:b}, \text{ and} \\
        \tau^{\rcn}_{t} + \tau^{\rcn}_{\bm{c},t} + \adj{E^{\rcn}|_{t-1}} \lambda^{\rcn}_{t} &\geqas \phi(x_t),\label{eq:sdp-cnst-ref:c}
        \end{alignat}
\end{subequations}
with $b = (R_w\hat{C}\trans{R_w} \pm\beta I)$ (as in \cref{ex:moment-dual}) and $\adj{E^{\rcn}|_{t-1}} \lambda^{\rcn}_{t} = \trans{\hm{w}_{t-1}} (\Lambda^{\rcn}_{t} - \uLambda^{\rcn}_{t}) \hm{w}_{t-1}$.}
We applied \cref{ex:robust-dual} for $\Pf^{t-1}$ and \cref{ex:moment-dual} for $\amb$. Note that \eqref{eq:sdp-cnst-ref:a} is a linear constraint. 
Since $\adj{E^{\rcn}|_{t-1}} \lambda^{\rcn}_{t}$ is quadratic,
\cref{eq:sdp-cnst-ref:b,eq:sdp-cnst-ref:c} are of similar structure to the constraint of \eqref{eq:sdp-cost}. 
That is, a quadratic inequality that should hold for all $\bm{w}_{:t-1} \in \W^{t-1}$. 
\update*{Hence, we} can once again apply the approximate S-Lemma followed by a Schur complement, resulting in a LMI.

It is clear that \eqref{eq:ocproblem:c} is equivalent to
\begin{equation}
        \trans{\bm{w}} \trans{\begin{bmatrix}
                \bm{H}_{N, :} & {h}_N \\ 0 & 1
        \end{bmatrix}} \begin{bmatrix}
                G_f & g_f \\ \trans{g}_f & \gamma_f
        \end{bmatrix}  \begin{bmatrix}
                \bm{H}_{N, :} & {h}_N \\ 0 & 1
        \end{bmatrix} \bm{w} \leq 0, \label{eq:sdp-terminal}
\end{equation}
for all $\bm{w} \in \W^N$. We then apply the approximate S-Lemma, followed by a Schur complement. 

Following these steps, we can add the LMI reformulations of \cref{eq:sdp-cnst-ref:a}--\eqref{eq:sdp-terminal} to 
the SDP formulation of \eqref{eq:sdp-cost}. The resulting problem conservatively approximates \eqref{eq:ocproblem}. 
See \ilarxiv{\cref{app:sdp}}\ilpub{\cite[\S5]{Arxiv}} for the full expressions and resulting SDP.
\end{pub}%
\begin{arxiv}%
We consider constraints, $\mathcal{r}^t_{\amb}$ with ambiguity $\amb = \amb_{\beta}(\smomenth)$, 
satisfying \cref{asm:risk}. The random variable $\phi(x_t)$ in \eqref{eq:ocproblem:b} equals
\begin{equation} \label{eq:risk-constraint-primal-var}
        \trans{\hm{x}}_t \begin{bmatrix}
                G & g \\ \trans{g} & \gamma
        \end{bmatrix} \hm{x}_t, \quad \text{with } \hm{x}_t = \begin{bmatrix}
                \bm{H}_{t, :t-1} & {h}_t \\ 0 & 1
        \end{bmatrix} \bm{\hm{w}}_{:t-1}
\end{equation}
and is thus quadratic in $\bm{w}_{:t-1}$. By construction, \eqref{eq:ocproblem:b} is of type \eqref{eq:ravar} with core ambiguity $\Pf^{t-1} \times \amb$.
Both ambiguity factors are conic by \cref{ex:robust-dual} and \cref{ex:moment-dual}. We now use $\mu^t = \mu^{t-1} \times \mu_{t-1}$. Then, 
similarly to before, by \cref{lem:cr-dual-prod}, $\Pf^{t-1} \times \amb$ is conic for
\begin{equation*}
        E^{\rcn}\mu^t =E^{\rcn}|_{t-1} \mu^t  = (\<\pm (R_w \hm{w}_{t-1} \trans{\hm{w}_{t-1}} \trans{R}_w ), \mu^t\>),
\end{equation*}
$b^{\rcn} = (R_w \smomenth \trans{R_w} \pm \beta I)$ and $\K^{\rcn} = \psd{n_w+1} \times \psd{n_w+1}$. 

Using the duality result of \cref{lem:ravar-conic} shows that \eqref{eq:ocproblem:b} holds if there exists a
$\tau_t^{\rcn}, \tau_{\bm{c}, t}^{\rcn} \in \Re$,
$\Lambda^{\rcn}_{t}, \uLambda^{\rcn}_{t} \sgeq 0$ and $\lambda^{\rcn}_{t} = (\Lambda^{\rcn}_{t}, \uLambda^{\rcn}_{t})$ such that
\begin{subequations} \label{eq:sdp-cnst-ref}
        \begin{alignat}{1}
        \tau^{\rcn}_t + \alpha^{-1} \left(\tau_{\bm{c}, t}^{\rcn}  + b^{\rcn}  \cdot \lambda^{\rcn}_t \right) &\leq 0 \label{eq:sdp-cnst-ref:a},\\
        \tau^{\rcn}_{\bm{c}, t} + \adj{E^{\rcn}|_{t-1}} \lambda^{\rcn}_{t} &\geqas 0\label{eq:sdp-cnst-ref:b}, \text{ and} \\
        \tau^{\rcn}_{t} + \tau^{\rcn}_{\bm{c}, t} + \adj{E^{\rcn}|_{t-1}} \lambda^{\rcn}_{t} &\geqas \phi(x_t).\label{eq:sdp-cnst-ref:c}
        \end{alignat}
\end{subequations}
We applied \cref{ex:robust-dual} for $\Pf^{t-1}$ and \cref{ex:moment-dual} for $\amb$. Note that \eqref{eq:sdp-cnst-ref:a} is a linear constraint. 
Since $\adj{E^{\rcn}|_{t-1}} \lambda^{\rcn}_{t} = \trans{\hm{w}_{t-1}} \trans{R_w} (\Lambda^{\rcn}_{t} - \uLambda^{\rcn}_{t}) R_w \hm{w}_{t-1}$,
\cref{eq:sdp-cnst-ref:b,eq:sdp-cnst-ref:c} are of similar structure to \eqref{eq:rme-primal-ref:b}. 
That is, a quadratic inequality that should hold for all $\bm{w}_{:t-1} \in \W^{t-1}$. 
Hence we can once again apply the approximate S-Lemma followed by a Schur complement, resulting in a LMI.

Specifically, let $\Delta^\rcn_t \dfn \trans{R}_w (\Lambda_t^\rcn - \uLambda_t^\rcn) R_w$. Then,
\begin{equation*}
        \adj{E^{\rcn}|_{t-1}} \lambda_t^{\rcn} = \trans{\hm{\bm{w}}_{:t-1}}\begin{bmatrix}
                & \\ && \matpart{\Delta_t^\rcn} & \vecpart{\Delta_t^\rcn}  \\ 
                && \trans{\vecpart{\Delta_t^\rcn}} &  \cstpart{\Delta_t^\rcn}
        \end{bmatrix}\hm{\bm{w}}_{:t-1}.
\end{equation*}

Following a similar procedure to before, we apply the approximate S-Lemma for both \eqref{eq:sdp-cnst-ref:b} and \eqref{eq:sdp-cnst-ref:c}.
Specifically, \eqref{eq:sdp-cnst-ref:b} holds if there exists some $s^{\rcn}_{1,t} \in \Re_+^{N-1}$ such that
\begin{align}
        \begin{bmatrix}
                [s^{\rcn}_{1,t}]_0 I \\ & \ddots \\ && [s^{\rcn}_{1,t}]_{t-1} I \\ &&& 
                \matpart{\Delta_t^\rcn}  & \vecpart{\Delta_t^\rcn} \\ &&& \trans{\vecpart{\Delta_t^\rcn}} &  \cstpart{\Delta_t^\rcn} + \tau_{\bm{c}, t}^{\rcn} - \ssum_{i=0}^{t-1}[s^{\rcn}_{1,t}]_{i} r^2
        \end{bmatrix} &\sgeq 0,  \label{eq:sdp-cnst-quad-a-ref} 
\end{align}
Meanwhile, \eqref{eq:sdp-cnst-ref:c} holds if there exists some $P^\rcn_t \in \psd{tn_w + 1}$ and $s^{\rcn}_{2,t} \in \Re_+^{N-1}$ such that
\begin{equation}
        \begin{bmatrix}
                [s^{\rcn}_{2,t}]_0 I \\ & \ddots \\ && [s^{\rcn}_{2,t}]_{t-1} I \\ &&& 
                \matpart{\Delta_t^\rcn} & \vecpart{\Delta_t^\rcn} \\ &&& \trans{\vecpart{\Delta_t^\rcn}} &  \cstpart{\Delta_t^\rcn} +  \tau_{\bm{c}, t}^{\rcn} + \tau_{t}^{\rcn}  - \ssum_{i=0}^{t-1}[s^{\rcn}_{2,t}]_{i} r^2
        \end{bmatrix} \sgeq P^{\rcn}_t, \label{eq:sdp-cnst-quad-b-ref}                
\end{equation}
where, analogously to $P^{\rme}$, the matrix $P^{\rcn}_t$ should satisfy
\begin{align}
        P_t^{\rcn} &\sgeq \trans{\begin{bmatrix}
                \bm{H}_{t, :t-1} & {h}_t \\ 0 & 1
        \end{bmatrix}} \begin{bmatrix}
                G & g \\ \trans{g} & \gamma
        \end{bmatrix}\begin{bmatrix}
                \bm{H}_{t, :t-1} & {h}_t \\ 0 & 1
        \end{bmatrix}\nonumber\\
        &= \trans{\begin{bmatrix}
                \bm{H}_{t, :t-1} & {h}_t
        \end{bmatrix}} G \begin{bmatrix}
                \bm{H}_{t, :t-1} & {h}_t
        \end{bmatrix} +  \begin{bmatrix}
                0 & \trans{\bm{H}_{t, :t-1}} g \\ \trans{g} \bm{H}_{t, :t-1} & 2 \trans{{h}_t} g + \gamma
        \end{bmatrix}. \label{eq:pre-schur-rcn}
\end{align}
We then apply a Schur complement to find an equivalent LMI for \eqref{eq:pre-schur-rcn}:
\begin{equation}  \label{eq:sdp-cnst-quad-b-aux}
        \begin{bmatrix}
                P_t^{\rcn} - \begin{bmatrix}
                        0 & \trans{\bm{H}_{t, :t-1}} g \\ \trans{g} \bm{H}_{t, :t-1} & 2 \trans{{h}_t} g + \gamma
                \end{bmatrix} & 
                        \trans{\begin{bmatrix}
                                \bm{H}_{t, :t-1} & {h}_t
                        \end{bmatrix}} G^{1/2} \\
                G^{1/2} \begin{bmatrix}
                        \bm{H}_{t, :t-1} & {h}_t 
                \end{bmatrix} & I
        \end{bmatrix}\sgeq 0.
\end{equation}

Constraint \eqref{eq:sdp-cnst-ref:a} is equivalent to
\begin{equation} \label{eq:sdp-cnst-lin-ref}
        \tau_{t}^{\rcn} + \alpha^{-1} \left(\tau_{\bm{c}, t}^{\rcn} + \tr[\Lambda_t^{\rcn} (R_w \hat{C} \trans{R_w} + \beta I)] + \tr[\uLambda_t^{\rcn} (R_w \hat{C} \trans{R}_w - \beta I)] \right) \leq 0.
\end{equation}

We conclude by reformulating the terminal constraint \eqref{eq:ocproblem:c}, which is equivalent to
\begin{equation}
        \trans{\bm{w}} \trans{\begin{bmatrix}
                \bm{H}_{N, :} & {h}_N \\ 0 & 1
        \end{bmatrix}} \begin{bmatrix}
                G_f & g_f \\ \trans{g}_f & \gamma_f
        \end{bmatrix}  \begin{bmatrix}
                \bm{H}_{N, :} & {h}_N \\ 0 & 1
        \end{bmatrix} \bm{w} \leq 0, \label{eq:sdp-terminal}
\end{equation}
for all $\bm{w} \in \W^N$. We apply the approximate S-Lemma, followed by a Schur complement. 
Specifically \eqref{eq:sdp-terminal} holds, if there exists some $s_f \in \Re^{N-1}_+$ such that
\begin{equation} \label{eq:sdp-terminal-ref}
        \begin{bmatrix}
                \begin{bmatrix}
                        \diag(s_f) \otimes I_{n_w} & -\trans{\bm{H}_{N, :}} g_f \\ -\trans{g}_f \bm{H}_{N,:} & -\ssum_{i=0}^{N-1} [s_f]_i r^2 -2 \trans{{h}_N} g_f - \gamma_f
                \end{bmatrix} & 
                        \trans{\begin{bmatrix}
                                \bm{H}_{N, :} & {h}_N
                        \end{bmatrix}} G^{1/2}_f \\
                G^{1/2}_f \begin{bmatrix}
                        \bm{H}_{N, :} & {h}_N
                \end{bmatrix} & I
        \end{bmatrix}\sgeq 0.
\end{equation}

To summarize, we minimize \eqref{eq:cost-cost} subject to \eqref{eq:slemma-cost} and \eqref{eq:schur-cost} (replacing the cumulative stage cost \eqref{eq:ocproblem:cost}). 
Then add constraints \eqref{eq:sdp-cnst-quad-a-ref}, \eqref{eq:sdp-cnst-quad-b-ref}, \eqref{eq:sdp-cnst-quad-b-aux}, \eqref{eq:sdp-cnst-lin-ref} for each $t \in \N_{1:N-1}$ 
(replacing the chance constraints \eqref{eq:ocproblem:b}).
The terminal constraint \eqref{eq:ocproblem:c} is then replaced by \eqref{eq:sdp-terminal-ref}. The resulting minimization problem is a SDP. 
\end{arxiv}

\begin{pub}
        \finalpage{}   
\end{pub}

\subsection{Numerical Results}
We consider two experiments: one where data is gathered offline and one where data is gathered online.
In both settings the controller acts more optimally when more data is available without affecting recursive feasibility. 
\paragraph{Offline learning}
Consider the linear system
\begin{equation*}
        x_{t+1} = \mat{
                0.9 & 0.2 \\ 0 & 0.8
        } x + \mat{0.1 \\ 0.05} u_t + \mat{0.5 & 0 \\ 0 & 0.1} w_t,
\end{equation*}
with $w_t$ distributed as a truncated Gaussian with covariance $\Sigma = 0.01^2 I$ and $\nrm{w_t}_2 \leq 0.15$ and $N = 8$. For the cost, we take 
$Q = \diag(1, 10), R = 10$ and $Q_f$ and $K_f$ and the solution to the discrete algebraic Riccati equation and the corresponding optimal controller respectively.
For the state constraints we have $\alpha = 0.2$ and $\phi(x) = [x]_1 - 1.5 \leq 0$,
where $[x]_1$ denotes the first element of $x$. For the terminal constraint we have $\psi(x) = \trans{x} Q_f x - 9.524 \leq 0$, which 
is RPI for $\pi_f(x) = K_f x$ and $x_0 = (1.75, 2)$.

We investigate sample counts $M$ between $10$ and $10^6$. For each sample count we gather an offline data set and 
compute $\amb$ from \cref{lem:dd-moment} for $\delta = 0.05$. Then we propagate the dynamics, with the receding horizon controller \eqref{eq:ocproblem} for $T = 15$ time steps, cumulating the stage cost 
for each step. This experiment is repeated $25$ times. The range of resulting costs is shown in \cref{fig:regret}. For $M = 10$ and $M = 10^6$,
closed-loop state trajectories are depicted.
The cost initially remains approximately constant, since the moment constraint is inactive due to insufficient samples.
Afterwards, the controller gains confidence, moving closer to the state constraint when more data is available, resulting in a lower cost.


\begin{figure}
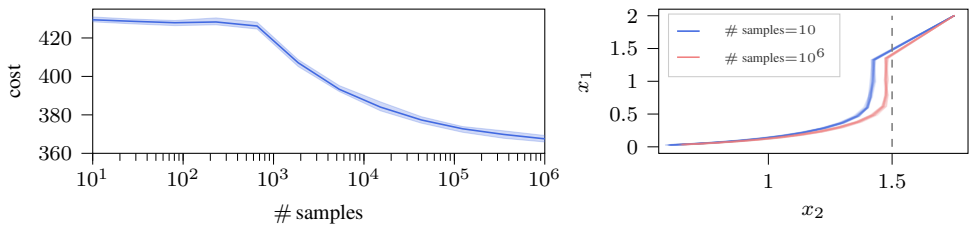

        \centering

\pgfplotsset{%
legend image code/.code={
    \draw[mark repeat=2,mark phase=2]
    plot coordinates {
    (0cm,0cm)
    (0.15cm,0cm)        
    (0.3cm,0cm)         
    };%
    }%
}

\begin{tikzpicture}
    \tikzstyle{every node}=[font=\footnotesize]
    \definecolor{color0}{rgb}{0.254901960784314,0.411764705882353,0.882352941176471}
    \definecolor{color1}{rgb}{0.941176470588235,0.501960784313725,0.501960784313725}

    \begin{axis}[
            name=plot_cost,
            ticklabel style={font=\footnotesize},
            width=0.6\linewidth,
            height=3.5cm,
            log basis x={10},
            log basis y={10},
            tick align=outside,
            tick pos=left,
            x grid style={white!69.0196078431373!black},
            xlabel={$\# \, \text{samples}$},
            xmin=10, xmax=1000000,
            xmode=log,
            xtick style={color=black},
            y grid style={white!69.0196078431373!black},
            xticklabel shift=-2pt,
            ylabel={cost},
            ymin=360, ymax=435,
            ylabel near ticks,
            xlabel near ticks,
            ytick style={color=black}
        ]
        \input{assets/hoeffding-offline-cost.tex}
    \end{axis}

    \begin{axis}[
            name=plot_states,
            at=(plot_cost.right of south east), anchor=left of south west,
            ticklabel style={font=\footnotesize},
            width=0.45\linewidth,
            height=3.5cm,
            legend cell align={left},
            legend style={
                fill opacity=0.8,
                draw opacity=1,
                text opacity=1,
                at={(0.03,0.97)},
                anchor=north west,
                draw=white!80!black
            },
            tick align=outside,
            ylabel near ticks,
            xlabel near ticks,
            tick pos=left,
            x grid style={white!69.0196078431373!black},
            xlabel={\(\displaystyle x_2\)},
            xmin=0.556916745705575, xmax=1.80681348829973,
            xtick style={color=black},
            y grid style={white!69.0196078431373!black},
            ylabel={\(\displaystyle x_1\)},
            ymin=-0.1, ymax=2.1,
            ytick style={color=black}
        ]
        
        \addplot [semithick, color0, opacity=1]
        table {%
        1.75 2
        1.75 2
        };
        \addlegendentry{$\scriptscriptstyle \# \, \text{samples} = 10$}
        
        \addplot [semithick, color1, opacity=1]
        table {%
        1.75 2
        1.75 2
        };
        \addlegendentry{$\scriptscriptstyle \# \, \text{samples} = 10^6$}
        \input{assets/hoeffding-offline-states.tex}
    \end{axis}
    
    \end{tikzpicture}
    \vspace{-0.75em}
        \caption{(left) Range of closed-loop costs when using the matrix-hoeffding bound to construct the ambiguity set 
        for $T = 15$; (right) state trajectories for $M = 10, 10^6$ offline samples.\vspace{-0.5em}} \label{fig:regret}
\end{figure}


\paragraph{Online learning}
We illustrate how our framework can be used to construct controllers that learn online, guaranteeing recursive feasibility. 
Consider the scalar linear system
\begin{equation*}
        x_{t+1} = x_t + u_t + w_t,
\end{equation*}
with $w_t$ distributed as a truncated Gaussian with variance $\sigma^2 = 0.05^2$, expectation $\E[w_t] = 0.05$ 
and $|w_t| \leq 1.0$. Let $Q = 1$, $R = 0.1$, $N=5$, $\alpha=0.2$ and
$\phi(x) = 1 - x \leq 0$. We use $\pi_f(x) = -0.9 x + 2.375$, $\psi(x) = (x - 2.375)^2 - 1.375^2 \leq 0$
and $Q_f$ the solution to the discrete algebraic Riccati equation. To implement learning, the ambiguity update 
scheme described in \ilarxiv{\cref{app:ambig-update}}\ilpub{\cite[App.~D]{Arxiv}} is used (which constructs ambiguity from data based on \cref{lem:dd-moment} for $\delta=0.05$ and rejects an ambiguity set if it 
is not a subset of the previously accepted ambiguity). 

We consider three cases: \emph{(i)} robust, with $\amb = \Pf$ and no learning; \emph{(ii)} online, starting from $M = 10$ samples
and then updating the ambiguity online while running the controller; and \emph{(iii)} offline, where $M = 10\,000$ initial samples and no updates. 
\cref{fig:online} depicts the state trajectory, which approaches the constraint for larger $M$.

\begin{figure}
        \centering
        \input{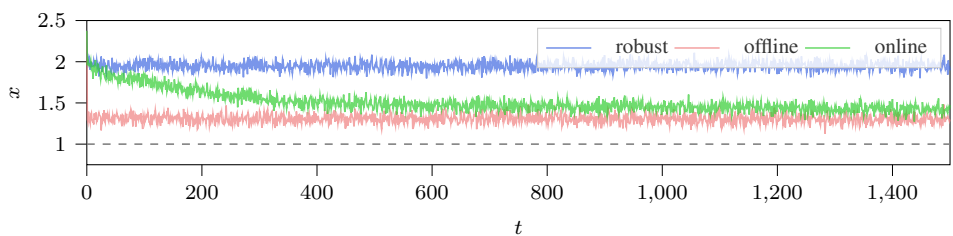}\vspace{-0.75em}
        \caption{\shorten*{State trajectories for robust; online-; and offline learning.}\vspace{-0.75em}} \label{fig:online}
\end{figure}

\printbibliography


\begin{arxiv}
\begin{appendix}
        \subsection{Proof of \cref{lem:dd-moment}} \label{app:dd-moment-proof}
        We begin by showing an auxiliary lemma, closely following the proof of \cite[Lemma.~2.1]{Rusten1992}.
        {\renewcommand{\thesection}{A}\begin{lemma} \label{lem:saddlepoint-matrix}
                Consider the following matrix, with $X \in \sym{d}$ and $x \in \Re^d$,
                \begin{equation*}
                        \hm{X} \dfn \begin{bmatrix} X & x \\ \trans{x} & 0 \end{bmatrix}.
                \end{equation*}
                Assume $\nrm{X}_2 \leq r^2_X$ and $\nrm{x}_2 \leq r_x$.
                Then $\nrm{X}_2 \leq (r_X^2 + \sqrt{r_X^4 + 4 r_x^2})/2$. 
        \end{lemma}}
        \begin{proof}
                Assume $(v, \gamma)$ is an eigenvector of $\hm{X}$ with associated eigenvalue $\lambda$.
                We assume without loss of generality that $\lambda \neq 0$. Then
                \begin{align}
                        X v + \gamma x &= \lambda v, \label{eq:saddlepoint-matrix-1} \\
                        \trans{x} v &= \lambda \gamma. \label{eq:saddlepoint-matrix-2}
                \end{align}
                Left multiplying \eqref{eq:saddlepoint-matrix-1} with $\trans{v}$ gives
                \begin{align}
                        &\trans{v} X v + \gamma \trans{x} v = \lambda \trans{v}v \nonumber \\
                        \Leftrightarrow\quad& p(\lambda) \dfn \lambda \trans{v} X v + \trans{v} (x \trans{x}) v - \lambda^2 \trans{v} v = 0,\label{eq:saddlepoint-matrix-3}
                \end{align}
                where equivalence holds due to \eqref{eq:saddlepoint-matrix-2} and the assumption $\lambda \neq 0$.

                \paragraph{(a) Assume $\lambda > 0$} in this case, from \eqref{eq:saddlepoint-matrix-3}, $\nrm{X}_2 \leq r_X^2$ and $\nrm{x}_2 \leq r_x$ we have 
                \begin{equation*}
                        (-r_X^2 \lambda - \lambda^2) \nrm{v}_2^2 \leq p(\lambda) \leq (r_X^2 \lambda + r_x^2 - \lambda^2) \nrm{v}_2^2.
                \end{equation*}
                Knowing that $\nrm{v}_2^2 \geq 0$ and $p(\lambda) = 0$, we should have $(-r_X^2 \lambda - \lambda^2) \leq 0$. This holds trivially since $\lambda \geq 0$.
                We should also have $(r^2_X \lambda + r^2_x - \lambda^2) \geq 0$, which results in a condition on $\lambda$ by 
                solving the quadratic inequality and taking the positive solution (since we assumed $\lambda \geq 0$). The resulting bound is
                $\lambda \leq (r_X^2 + \sqrt{r^4_X + 4 r_x^2})/2$.

                \paragraph{(b) Assume $\lambda < 0$} in this case we have 
                \begin{equation*}
                        (r_X^2 \lambda - \lambda^2) \nrm{v}_2^2 \leq p(\lambda) \leq (-r_X^2 \lambda + r^2_x - \lambda^2) \nrm{v}_2^2.
                \end{equation*}
                Following similar arguments as above, we see that the first inequality holds trivially, while the second holds iff $\lambda \geq -(r^2_X + \sqrt{r^4_X + 4r_x^2})/2$. 

                We have thus shown $-(r_X^2 + \sqrt{r_X^4 + 4r_x^2})/2 \leq \lambda \leq (r_X^2 + \sqrt{r_X^4 + 4 r_x^2})/2$, thereby proving the claimed result. 
        \end{proof}

        Using this lemma we can now proceed with the main proof.

        Let $\hm{X_i} \dfn R_w (\hm{w}_i \trans{\hm{w}_i} - \E_{\mu_\star}[\hm{w} \,\trans{\hm{w}}])  \trans{R_w}$ 
        and its eigenvalue decomposition $\hm{X_i} = V_i D_i \trans{V_i}$. 
        Clearly this matrix is structured as $\hm{X}$ in \cref{lem:saddlepoint-matrix}, where $X = w_i \trans{w}_i - \E_{\mu_\star}[w \trans{w}]$
        and $x = c r (w_i - \E_{\mu_\star}[w])$. We know that a.s., $-r^2 I \sleq -\E_{\mu_\star}[w \trans{w}] \sleq X \sleq w_i \trans{w_i} \sleq r^2 I$,
        hence $\nrm{X}_2 \leq r^2$. From a triangle inequality we also know that $\nrm{x}_2 \leq 2 c r^2$ a.s. Applying \cref{lem:saddlepoint-matrix} 
        then shows $\nrm{\hm{X}_i}_2 \leq r^2 (1 + \sqrt{1 + 16 c})/2 \nfd \bar{\lambda}$ a.s. Therefore the eigenvalues are bounded as $|d_i| \leq \bar{\lambda}$. 
                
        As such, we can produce a matrix sub-Gaussian bound:
        \begin{align*}
                \E[\exp(\theta \hm{X}_i)] &= \E[V_i \exp(\theta D_i) \trans{V_i}] \\
                                     &\labelrel\sleq{step:convexity} \E\left[\tfrac{(\bar{\lambda} I + \hm{X_i})}{2\bar{\lambda}} \exp(\theta \bar{\lambda}) + \tfrac{(\bar{\lambda} I - \hm{X_i})}{2\bar{\lambda}}\exp(-\theta \bar{\lambda}) \right]\\
                                     &\labelrel={step:zero-mean} \tfrac{\exp(\theta \bar{\lambda}) + \exp(-\theta \bar{\lambda})}{2}I 
                                     = \cosh(\theta \bar{\lambda}) I \sleq \exp(\theta^2 \bar{\lambda}^2/2) I,
        \end{align*}
        where \ref{step:convexity} follows from convexity and $V_i \trans{V_i} = I$, and \ref{step:zero-mean} from $\E[\hm{X_i}] = 0$.
        
        From independence and \cite[Lem.~3.5.1]{Tropp2015} we have, $\E[\tr[\exp(\ssum_{i=0}^{M-1}  \hm{X_i})]] \leq d \exp(M \theta^2 \bar{\lambda}^2/2)$,
        with $d \dfn n_w + 1$. 
        This allows us to apply \cite[Prop.~3.2.1]{Tropp2015}:
        \begin{align*}
                \prob[\lambda_{\mathrm{max}}(\ssum_{i=0}^{M-1}  \hm{X_i}) \geq \beta] &\leq \inf_{\theta > 0} d \exp(-\theta \beta + M \theta^2 \bar{\lambda}^2/2) 
                = d\exp(-\beta^2/(2M \bar{\lambda}^2)).
        \end{align*}
        Similarly bounding $\lambda_{\mathrm{min}}$ and applying a union bound gives
        \begin{equation*}
                \prob[\nrm{\ssum_{i=0}^{M-1} \hm{X_i}} \geq \beta] \leq 2d\exp(-\beta^2/(2M \bar{\lambda}^2)).
        \end{equation*}
        This is equivalent to 
        \begin{equation*}
                \prob[\nrm{\ssum_{i=0}^{M-1} \hm{X_i}/M} \geq \beta] \leq 2d\exp(-M\beta^2/(2 \bar{\lambda}^2)),
        \end{equation*}
        by a change of variable for $\beta$. Setting the right-hand-side equal to $\delta$ and solving for $\beta$ gives the value in \cref{lem:dd-moment},
        showing $\prob[\nrm{\smomenth - \smoment}_2 \leq \beta] \geq 1-\delta$, implying \eqref{eq:amb_conf}.\qed{}

        \subsection{Continuity of iterated integration} \label{app:continuity}
        We argue why $h(\bm{w}_{:t-2})$ in the proof of \cref{thm:rf} is a measurable, continuous random variable. It is given by
        \begin{equation*}
                h(\bm{w}_{:t-2}) = \int_{\W} z_{\tau}(\bm{w}_{:t-2}, w_{t-1}) \di \mu(w_{t-1}),
        \end{equation*}
        for continuous $z_{\tau} \in \Z^t$. Measurability follows by Tonelli's theorem \cite[Thm.~5.2.1]{Athreya2006}, compactness of $\W$ and continuity of $z_{\tau}$ (implying $\sigma$-finiteness of the measure spaces and integrability).
        Meanwhile, continuity follows by an epsilon-delta argument \cite[\S9.3]{Bruckner2007}.
        Specifically, Jensen's inequality gives
        \begin{align}
                |h(\bm{w}_{:t-2})  - h(\bm{w}_{:t-2}')| &\leq \int_{\W} |z_{\tau}(\bm{w}_{:t-2}, w_{t-1}) - z_{\tau}(\bm{w}_{:t-2}', w_{t-1})| \di\mu_{t-1}(w_{t-1}), \nonumber\\
                                                        &\leq \max_{w_{t-1}\in\W} \left\{|z_{\tau}(\bm{w}_{:t-2}, w_{t-1}) - z_{\tau}(\bm{w}_{:t-2}', w_{t-1})| \right\}\label{eq:hcont:a}
        \end{align}
        for any $\bm{w}_{:t-2}, \bm{w}_{:t-2}' \in \W^{t-1}$. The second inequality holds since $|z_{\tau}(\bm{w}_{:t-2}, w_{t-1}) - z_{\tau}(\bm{w}_{:t-2}', w_{t-1})|$
        is continuous in $w_{t-1}$ by continuity of $z_{\tau}$, which enables the application of \cref{ex:robust-dual} 
        to replace integration over $\mu_{t-1} \in \Pf(\W)$ by a maximization over $w_{t-1} \in \W$. 

        From continuity of $z_{\tau}$ we have
        \begin{equation*}
                \forall \varepsilon \geq 0, \exists \delta \geq 0 \text{ s.t. } \nrm{\bm{w} - \bm{w}'} \leq \delta \Rightarrow |z(\bm{w}) - z(\bm{w}')| \leq \epsilon. 
        \end{equation*}
        So picking $\bm{w} = (\bm{w}_{:t-2}, w_{t-1})$ and $\bm{w}' = (\bm{w}_{:t-2}', w_{t-1})$ and plugging into \eqref{eq:hcont:a} gives
        \begin{equation} \label{eq:hcont:b}
                \nrm{\bm{w}_{:t-2} - \bm{w}_{:t-2}'} \leq \delta \Rightarrow |h(\bm{w}_{:t-2})  - h(\bm{w}_{:t-2}')| \leq \varepsilon. 
        \end{equation}
        Hence, for any $\varepsilon \geq 0$ we have a $\delta\geq 0$ s.t. \eqref{eq:hcont:b} holds, implying continuity of $h$. \qed{}

        \subsection{Affine disturbance feedback notation} \label{app:aff-dist-not}
        Let $N_x \dfn (N+1)n_x$, $N_w \dfn Nn_w$ and $N_u \dfn Nn_u$. Following \cite{Goulart2005}, we consider the matrices,
        $\bm{A} \in \Re^{N_x \times n_x}$, $\bm{B} \in \Re^{N_x \times N_u}$,
        $\bm{E} \in \Re^{N_x \times N_w}$ and $\bm{F} \in \Re^{N_u \times N_w}$
        such that
        {\allowdisplaybreaks{}%
        \begin{align} \label{eq:causality constraints}
                \bm{A} &\dfn \smallmat{
                        I_{n_x} \\ A \\ A^2 \\ \vdots \\ A^N
                }, & 
                \bm{B} &\dfn \smallmat{
                        0 & 0 & \dots & 0\\
                        B & 0 & \dots & 0\\
                        AB & B & \dots & 0 \\
                        \vdots & \vdots & \ddots & \vdots \\
                        A^{N-1}B & A^{N-2}B & \dots & B
                }, \\
                \bm{E} &\dfn \smallmat{
                        0 & 0 & \dots & 0 \\
                        E & 0 & \dots & 0 \\
                        AE & E & \dots & 0  \\
                        \vdots & \vdots & \ddots &  \vdots \\
                        A^{N-1}E & A^{N-2}E & \dots & E
                }, &
                \bm{F} &\dfn \smallmat{
                        0 & 0 & \dots & 0  \\
                        F_{1,0} & 0 & \dots & 0 \\
                        \vdots & \vdots & \ddots & \vdots \\
                        F_{N-1, 0} & \dots & F_{N-1, N-2} & 0
                },
        \end{align}%
        }%
        $\bm{Q} = \diag(Q, \dots, Q, Q_f)$ and
        $\bm{R} = \diag(R, \dots, R)$. 

        \subsection{Updating the ambiguity} \label{app:ambig-update}
        In this section we examine how \cref{asm:amb} can be guaranteed in practice. More specifically, assume we have 
        $M$ offline samples used to construct $\amb_0$ according to \cref{lem:dd-moment}. At time step $t$ we then will have
        gathered $M + t$ samples, which can be used to construct $\hat{\amb}_t$. We then use the update rule
        \begin{equation} \label{eq:cont-suff-1}
                \amb_{t+1} = \begin{cases}
                        \hat{\amb}_{t+1}, & \text{if } \hat{\amb}_{t+1} \subseteq \amb_t \\
                        \amb_t, & \text{else}.
                \end{cases}
        \end{equation}
        Hence we only need to check $\hat{\amb}_{t+1} \subseteq \amb_t$. Letting $\smomenth_+$ ($\beta_+$), $\smomenth$ ($\beta$) denote the center (radius) of $\hat{\amb}_{t+1}$ and $\amb_t$
        respectively. Then a sufficient condition for $\hat{\amb}_{t+1} \subseteq \amb_t$ is,
        \begin{equation*}
                \nrm{R_w (\smomenth - \smoment) \trans{R_w} }_2 \leq \beta, \quad \forall \smoment \sgeq 0, \text{ s.t., } \nrm{R_w (\smomenth_+ - \smoment) \trans{R_w}}_2 \leq \beta_+.
        \end{equation*}
        This condition is only sufficient since the bounded support would imply $\tr([\smoment]_{:n_w, :n_w}) \leq r^2$, where $[\smoment]_{:n_w, :n_w}$ is a matrix containing the first $n_w$ rows and columns of $\smoment$ (i.e., $\E[w \trans{w}]$). 
        The condition that $\mu$ is a probability measure implies $[\smoment]_{n_w+1, n_w+1} = 1$, where $[\smoment]_{n_w+1, n_w+1}$ denotes the bottom-right element of $\smoment$. Based on numerical experiments however, dropping these 
        constraints does not result in significant conservativeness. 

        Condition \eqref{eq:cont-suff-1} is not trivial to check however. It requires maximizing the convex function $\nrm{R_w (\smomenth - \smoment) \trans{R_w}}_2$ over a convex set, which is NP hard in general.
        We however have some structure we can exploit. Note how
        \begin{equation*}
                \nrm{R_w (\smomenth - \smoment) \trans{R}_w }_2 \leq \beta \qquad \Leftrightarrow \qquad \trans{\begin{bmatrix}
                        \smoment \\ I
                \end{bmatrix}} \begin{bmatrix}
                        -I & \smomenth \\ \smomenth & \beta^2 (\trans{R_w} R_w)^{-2} - \smomenth^2
                \end{bmatrix}\begin{bmatrix}
                        \smoment \\ I
                \end{bmatrix} \sgeq 0.
        \end{equation*}
        This, should hold for all $\smoment \sgeq 0$ such that
        \begin{equation*}
                \trans{\begin{bmatrix}
                        \smoment \\ I
                \end{bmatrix}} \begin{bmatrix}
                        -I & \smomenth_+ \\ \smomenth_+ & \beta_+^2 (\trans{R_w} R_w)^{-2}  - \smomenth_+^2
                \end{bmatrix}\begin{bmatrix}
                        \smoment \\ I
                \end{bmatrix} \sgeq 0.
        \end{equation*}
        A sufficient condition for which is that,
        \begin{equation*}
                \trans{v} \begin{bmatrix}
                        -I & \smomenth \\ \smomenth & \beta^2 (\trans{R_w} R_w)^{-2} - \smomenth^2
                \end{bmatrix} v \geq 0,
        \end{equation*}
        for all $v \in \Re^{2(n_w+1)}$, such that
        \begin{equation*}
                \trans{v} \begin{bmatrix}
                        -I & \smomenth_+ \\ \smomenth_+ & \beta_+^2 (\trans{R_w} R_w)^{-2}  - \smomenth_+^2
                \end{bmatrix} v \geq 0 \text{ and }
                \trans{v} \begin{bmatrix}
                        & I \\ I
                \end{bmatrix} v \geq 0.
        \end{equation*}
        In fact this condition is sufficient and necessary for \eqref{eq:cont-suff-1} (cf. \cite{VanWaarde2020}). 
        Applying the Approximate S-Lemma \cite[Thm.~B.3.1.]{Ben-Tal2009} (which is exact in this case) allows us to formulate this final condition as a LMI. 
\end{appendix}
\end{arxiv}

\end{document}